\documentclass[a4paper,12pt]{article}
\usepackage[utf8]{inputenc}
\usepackage[english]{babel}
\usepackage{amsmath,amssymb,amsfonts,amsthm,textcomp,natbib}
\usepackage{graphicx}
\graphicspath{{./images/}}
\usepackage{listings}
\usepackage{tabularx}
\usepackage{fullpage}
\usepackage{enumitem}
\newtheorem{theorem}{Theorem}
\newtheorem{definition}{Definition}

\newtheorem{lemma}{Lemma}
\newtheorem{corollary}{Corollary}

\makeatletter
\@addtoreset{proofpart}{theorem}
\makeatother
\usepackage{tikz}
\usetikzlibrary{positioning,arrows,shapes,shadows,snakes}
\usetikzlibrary{intersections,decorations.markings,calc,positioning}

\usepackage{soul}

\makeatletter
\newcommand\RedeclareMathOperator{%
  \@ifstar{\def\rmo@s{m}\rmo@redeclare}{\def\rmo@s{o}\rmo@redeclare}%
}
\newcommand\rmo@redeclare[2]{%
  \begingroup \escapechar\m@ne\xdef\@gtempa{{\string#1}}\endgroup
  \expandafter\@ifundefined\@gtempa
     {\@latex@error{\noexpand#1undefined}\@ehc}%
     \relax
  \expandafter\rmo@declmathop\rmo@s{#1}{#2}}
\newcommand\rmo@declmathop[3]{%
  \DeclareRobustCommand{#2}{\qopname\newmcodes@#1{#3}}%
}
\@onlypreamble\RedeclareMathOperator
\makeatother

\newcommand{\SE}[2][]{\mathbf{SE}_{#2}^{#1}}
\newcommand{\NW}[2][]{\mathbf{NW}^{#1}_{#2}}

\DeclareMathOperator{\cgeq}{\succcurlyeq}
\DeclareMathOperator{\cleq}{\preccurlyeq}
\DeclareMathOperator{\cgt}{\succ}
\DeclareMathOperator{\clt}{\prec}
\DeclareMathOperator{\R}{\mathbf{R}}
\RedeclareMathOperator{\S}{\mathbf{S}}
\DeclareMathOperator{\E}{\mathbf{E}}

\begin{document}
\title{Conjoint axiomatization of the Choquet integral for heterogeneous product sets}
\author{Mikhail Timonin}
\maketitle

\begin{abstract}
  We propose an axiomatization of the Choquet integral model for the general case of a heterogeneous product set
  $X = X_1 \times \ldots \times X_n$. In MCDA elements of $X$ are interpreted as alternatives, characterized by criteria taking values from
  the sets $X_i$. Previous axiomatizations of the Choquet integral have been given for particular cases $X = Y^n$ and
  $X = \mathbb{R}^n$. However, within multicriteria context such identicalness, hence commensurateness, of criteria cannot be assumed a
  priori. 
  This constitutes the major difference of this paper from the earlier axiomatizations. In particular, the notion of ``comonotonicity'' cannot be
  used in a heterogeneous structure, as there does not exist a ``built-in'' order between elements of sets $X_i$ and $X_j$. However, such
  an order is implied by the representation model. Our approach does not assume commensurateness of criteria. We construct the
  representation and study its uniqueness properties.
\end{abstract}

\section{Introduction}
\label{sec:introduction}

The Choquet integral is widely used in decision analysis and, in particular, MCDA \cite{grabisch2008decade}, although its use is still
somewhat restricted due to both methodological problems and difficulties in practical implementation. Rank-dependent models first appeared
in the axiomatic decision theory in reply to the criticism of Savage's postulates of rationality \cite{savage1954foundations}. The renowned
Ellsberg paradox \cite{ellsberg1961risk} has shown that people can violate Savage's axioms and still consider their behaviour
rational. First models accounting for the so-called uncertainty aversion observed in this paradox appeared in the 1980s, in the work
\cite{quiggin1982theory} and others (see \cite{wakker1991additive-RO} for a review). One particular generalization of the expected utility
model (EU) characterized by Schmeidler \cite{schmeidler1989subjective} is the Choquet expected utility (CEU), where probability is replaced
by a non-additive set function (called capacity) and integration is performed using the Choquet integral.

Since Schmeidler's paper, various versions of the same model have been characterized in the literature
(e.g. \cite{gilboa1987expected,wakker1991additive}). CEU has gained some momentum in both theoretical and applied economic literature, being
used mainly for analysis of problems involving Knightian uncertainty. At the same time, rank-dependent models, in particular the Choquet
integral, were adopted in multiattribute utility theory (MAUT) \cite{keeney1993decisions}. Here the integral gained popularity due to the
tractability of non-additive measures in this context (see \cite{grabisch2008decade} for a review). The model permitted various preferential
phenomena, such as criteria interaction, which were impossible to reflect in the traditional additive models.

The connection between MAUT and decision making under uncertainty has been known for a long time. In the case when the number of states is
finite, which is assumed hereafter, states can be associated with criteria. Accordingly, acts correspond to multicriteria
alternatives. Finally, the sets of outcomes at each state can be associated with the sets of criteria values. However, this last transition
is not quite trivial. It is commonly assumed that the set of outcomes is the same in each state of the world
\cite{savage1954foundations,schmeidler1989subjective}. In multicriteria decision making the opposite is true. Indeed, consider preferences
of consumers choosing cars. Each car is characterized by a number of features (criteria), such as colour, maximal speed, fuel consumption,
comfort, etc. Apparently, sets of values taken by each criterion can be completely different from those of the others. In such context the
ranking stage of rank-dependent models, which in decision under uncertainty involves comparing outcomes attained at various states, would
amount to comparing colours to the level of fuel consumption, and maximal speed to comfort. Indeed, the traditional additive model
\cite{debreu1959topological,krantz1971foundation} only implies meaningful comparability of units between goods in the bundle, but not of
their absolute levels. However, in rank-dependent models such comparability seems to be a necessary condition.

We propose a representation theorem for the Choquet integral model in the MCDA context. Binary relation $\succcurlyeq$ is defined on a
heterogeneous product set $X = X_1 \times \ldots \times X_n$. In multicriteria decision analysis (MCDA), elements of the set $X$ are
interpreted as alternatives, characterized by criteria taking values from sets $X_i$. Previous axiomatizations of the Choquet integral model
have been given for the special cases of $X = Y^n$ (see \cite{kobberling2003preference} for a review of approaches) and $X = \mathbb{R}^n$
(see \cite{grabisch2008decade} for a review). One related result is the recent axiomatization of the Sugeno integral model
(\cite{greco2004axiomatic,bouyssou2009conjoint}). Another approach using conditions on the utility functions was proposed in
\cite{labreuche2012axiom}. The ``conjoint'' axiomatization of the Choquet integral for the case of a general $X$ was an open problem in the
literature. The crucial difference with the previous axiomatizations is that the notion of ``comonotonicity'' cannot be used in the
heterogeneous case, due to the fact that there does not exist a meaningful ``built-in'' order between elements of sets $X_i$. New axioms and
modifications of proof techniques had to be introduced to account for that.

Our first axiom shows, roughly, how the set $X$ can be partitioned into subsets based on properties necessary for existence of an additive
representation. The axiom (\textbf{A3}) we introduce is similar to the ``2-graded'' condition previously used for characterizing of MIN/MAX
and the Sugeno integral (\cite{greco2004axiomatic,bouyssou2009conjoint}). At every point $z \in X$ for every pair of coordinates $i,j \in N$
it is possible to build two ``rectangular cones'' - one made up of points from $X_i$ which are ``greater'' than $z_i$ and points from $X_j$
which are ``less'' than $z_j$, and the second for the opposite case. The axiom states that triple cancellation for $\cgeq$ restricted to
$i,j$ must then hold on at least one of these cones. This allows to partition $X$ into subsets by using intersection of such cones for
various pairs $i,j$.

The second property is that the additive representations on different subsets are interrelated, in particular ``trade-offs'' between
criteria values are consistent across partition elements both within the same dimension and across different ones. This is reflected by two
axioms (\textbf{A4, A5}), similar to the ones used in \cite{wakker1991additive} and \cite{krantz1971foundation} (section 8.2). One, roughly
speaking, states that triple cancellation holds across subsets, while the other says that ordering of intervals on any dimension must be
preserved when they are projected onto another dimension by means of equivalence relations. These axioms are complemented by a new condition
called bi-independence (\textbf{A6}) and weak separability (\textbf{A2}) \cite{bouyssou2009conjoint} - which together reflect the
monotonicity property of the integral, and also the standard essentiality, ``comonotonic'' Archimedean axiom and restricted solvability
(\textbf{A7,A8,A9}). Finally, $\cgeq$ is supposed to be a weak order (\textbf{A1}), and $X$ is order dense.


\section{Choquet integral in MCDA}
\label{sec:choquet-integral}

\begin{definition}
  Let $ N = \{1, \ldots, n \} $ be a finite set and $ 2^N $ its power set. Capacity (non-additive measure, fuzzy measure) is a set function
  $\nu:2^N\rightarrow \mathbb{R}_+$ such that:
  \begin{enumerate}
  \item $ \nu(\varnothing)=0 $;
  \item $ A \subseteq B \Rightarrow \nu(A)\leq\nu(B), \ \forall A,B \in 2^N $.
  \end{enumerate}
  In this paper, it is also assumed that capacities are normalized, i.e.\ $ \nu(N) = 1 $.
\end{definition}
\begin{definition}\label{def:choq-integr}
  The Choquet integral of a function $ f:N \rightarrow \mathbb{R} $ with respect to a capacity $ \nu $ is defined as
  \begin{equation*}
    C(\nu,f) = \int \limits_{0}^{\infty} \nu( \{ i \in N \colon f(i) \geq r \})dr + \int \limits_{-\infty}^{0} [\nu( \{ i \in N \colon f(i)
    \geq r \}) - 1]dr
  \end{equation*}
Denoting the range of $f:N \rightarrow \mathbb{R}$ as $\{f_1,\ldots,f_n\}$, the definition can be written down as:
   \begin{equation*}
    C(\nu,(f_1,\ldots,f_n)) =
    \sum \limits_{i=1}^{n} (f_{(i)}-f_{(i-1)})\nu(\{ j \in N \colon f_j \geq  f_{(i)} \} )
  \end{equation*}
  where $f_{(1)},\ldots,f_{(n)}$ is a permutation of $f_1,\ldots,f_n$ such that $f_{(1)} \leq f_{(2)} \leq \cdots \leq
  f_{(n)}$, and $f_{(0)}=0$. 
\end{definition}

On of the most useful tools for analysis of the capacity is the so-called M{\"o}bius transform. It's a linear transformation of the capacity
which is given by:
\begin{equation*}
  m(A)=\sum _{B\subset A}(-1)^{|A\setminus B|}\nu (B).
\end{equation*}  

The Choquet integral can be written in a very convenient form using the M{\"o}bius transform coefficients:
\begin{equation*}
  C(\nu,f) = \sum _{A \in N} m(A) \min_{i \in A} (f_i).
\end{equation*}

\subsection{The model}
\label{sec:model}

Let $\cgeq$ be a binary relation on the set $X = X_1 \times \ldots \times X_n$. $\cgt, \clt, \cleq, \sim, \not \sim$ are defined in the
usual way. In MCDA, elements of set $X$ are interpreted as alternatives characterized by criteria from the set $N = \{1,\ldots,n\}$. Set
$X_i$ contains criteria values for criterion $i$. We say that $\cgeq$ can be represented by a Choquet integral, if there exists a capacity
$\nu$ and functions $f_i: X_i \rightarrow \mathbb{R}$, called value functions, such that:
\begin{equation*}
  x \cgeq y \iff C(\nu,(f_1(x_1),\ldots,f_n(x_n)) \geq C(\nu,(f_1(y_1),\ldots,f_n(y_n)).
\end{equation*}

As seen in the definition of the Choquet integral, its calculation involves comparison of $f_i$'s to each other. It is not immediately
obvious how this operation can have any meaning in the MCDA decision framework. It is well-known that direct comparison of value functions
for various attributes is meaningless in the additive model \cite{krantz1971foundation} (recall that the origin of each value function can
be changed independently). In the homogeneous case $X = Y^n$ this problem is readily solved, as we have a single set of ``consequences'' $Y$
(in the context of decision making under uncertainty). The required order is either assumed as given \cite{wakker1991additive-RO} or is
readily derived from the ordering of ``constant'' acts $(y, \ldots, y)$ \cite{wakker1991additive}. Since there is a single ``consequence''
set, we also only have one value function $U:Y \rightarrow \mathbb{R}$, and thus comparing $U(y_i)$ to $U(y_j)$ is perfectly sensible, since
$U$ represents the order on the set $Y$. None of these methods can be readily applied in the heterogeneous case.

\subsection{Properties of the Choquet integral}
\label{sec:prop-choq-integr}

Below are given some important properties of the Choquet integral:
\begin{enumerate}
\item Functions $f:N \rightarrow \mathbb{R}$ and $g:N \rightarrow \mathbb{R}$ are comonotonic if for no $i,j \in N$ holds $f(i) > f(j)$ and
  $g(i) < g(j)$. For all comonotonic $f$ the Choquet integral reduces to a usual Lebesgue integral. In the finite case, the integral is
  accordingly reduced to a weighted sum.
\item Particular cases of the Choquet integral (e.g. \cite{grabisch2008decade}). 
  \begin{itemize}
  \item If $m(\{1\}) = \ldots = m(\{n\}) = 1$, then $C(\nu,(f_1,\ldots,f_n)) = \max(f_1,\ldots,f_n)$.
  \item If $m(N) = 1, m(A) = 0, A \neq N$, then $C(\nu,(f_1,\ldots,f_n)) = \min(f_1,\ldots,f_n)$.
  \item If $m(A) = 0 $, for all $A \subset N : |A| \geq 2 $, then $C(\nu,(f_1,\ldots,f_n)) = \sum _{i \in N} \nu(\{i\})f_i$
  \end{itemize}
\end{enumerate}

Property 1 states that the set $X$ can be partitioned into subsets corresponding to particular ordering of the value functions. There are
$n!$ such sets. Since the integral on each of the sets is reduced to a weighted sum, i.e. an additive representation, we should expect many
of the axioms of the additive conjoint model to be valid on this subsets. This is the intuition behind several of the axioms given in the
following section.


\section{Axioms and definitions}
\label{sec:axioms-definitions}

\begin{definition}
  Given $i,j \in N$, a relation $\cgeq$ on $X_1 \times \ldots \times X_n$ satisfies \emph{$ij$-triple cancellation (\textbf{ij-3C})}, if for all $a_i,b_i, c_i,
  d_i \in X_i$, $p_j,q_j,r_j,s_j \in X_j$, and all $z_{-ij} \in X_{-ij}$ holds:
    \begin{equation*}
      \left.
        \begin{aligned}
          a_ip_jz_{-ij} & \cleq b_iq_jz_{-ij}\\
          a_ir_jz_{-ij} & \cgeq b_is_jz_{-ij}\\
          c_ip_jz_{-ij} & \cgeq d_iq_jz_{-ij} 
        \end{aligned}
      \right\} \Rightarrow c_ir_jz_{-ij} \cgeq d_is_jz_{-ij}.
    \end{equation*}
\end{definition}


\begin{description}
    \item[A1 - Weak order.] $\cgeq$ is a weak order. 
    \item[A2 - Weak separability.]  For all $i$, if $a_ix_{-i} \cgt b_ix_{-i}$ for some $a_i,b_i \in X_i$,$x_{-i} \in X_{-i}$, then
      $a_iy_{-i} \cgeq b_iy_{-i}$ for all $y_{-i} \in X_{-i}$.
\end{description}
Note, that from this follows, that for any $a_i,b_i \in X_i$ either $a_ix_{-i} \cgeq b_ix_{-i}$ or $b_ix_{-i} \cgeq a_ix_{-i}$ for all $x_{-i} \in X_{-i}$. This allows to introduce the following definition:
\begin{definition}
  For all $a_i, b_i \in X_i$ define $\cgeq_i$ as $a_i \cgeq_i b_i \iff a_ix_{-i} \cgeq b_ix_{-i}$ for all $x_{-i} \in X_{-i}$. 
\end{definition}
%
%
\begin{definition}
  For any $z \in X$ define $\SE[z]{ij} = \{x_ix_jz_{-ij} \in X \colon x_i \cgeq_i z_i, z_j \cgeq_j x_j \}$, and $\NW[z]{ij}
  = \{x_ix_jz_{-ij} \in X \colon z_i \cgeq_i x_i, x_j \cgeq_j z_j\}$.
\end{definition}
\begin{description}
    \item[A3 - Coordinate Ordering Completeness.] For any $z \in X$, and all $i,j \in N$, $ij$-triple cancellation holds either on $\SE[z]{ij}$ or on $\NW[z]{ij}$.
\end{description}
This new property would allow us to divide $X$ into subsets without the need to use the notion of comonotonicity. We can introduce the following
binary relations:
\begin{definition}
  We write:  
  \begin{enumerate}
  \item $i \R^z j$ if $ij$-triple cancellation holds on the set $\SE[z]{ij}$.
  \item $i \S^z j$ if [NOT $j \R^z i$].
  \item $i \E^z j$ if [$i \R^z j$ AND $j \R^z i$]. 
  \end{enumerate}
\end{definition}
Note that $\R^z$ is complete (which is why we have called axiom \textbf{A3} ``Coordinate Ordering Completeness'') and $\S^z$ is partial.\footnote{if it is empty for all $z$, other axioms entail the existence of an
  additive representation on $X$} Since $N$ is finite, there is only a finite number of various partial orders $\S^z$, so
we can index them ($\S_a, \S_b,\ldots$) and drop the superscripts when not needed. Also, each of the partial orders $\S_k$ uniquely defines
the corresponding $\R_k$ - $i \R_k j$ if [NOT $j \S_k i$].

In contrast to the case with two variables, this property alone is not sufficient to construct a representation. Comparing value functions
for different attributes suggests some sort of transitivity. For example, $f_i(x_i) > f_j(x_j)$ and $f_j(x_j) > f_k(x_k)$ imply $f_i(x_i) >
f_k(x_k)$. The property we introduce is weaker - it is acyclicity.
\begin{description}
\item[A3-ACYCL - Coordinate Ordering Acyclicity.] For all $z \in X$, $\S^z$ is acyclic. In other words, 
  \begin{equation*}
     i \S^z j \S^z \ldots \S^z k \Rightarrow  i \R^z k.
  \end{equation*}
\end{description}
This axiom effectively defines how the set $X$ is partitioned. It is required for the Choquet integral representation to
exist. 


We also introduce the following notions:

\begin{definition}
  Define \emph{$\SE{ij}$} as a union of the following three sets:
  \begin{itemize}
  \item All $z \in X$ such that $i \R^{z} j$, if $z_i$ is not maximal and $z_j$ is not minimal;
  \item All $z \in X$ such that $z_i$ is maximal and for no $x_j, y_j \in X_j : z_j \cgeq_j x_j \cgeq_j y_j$ we have $j \R^{x_jz_{-j}} i$ and NOT $j
    \R^{y_jz_{-j}} i$;
  \item All $z \in X$ such that $z_j$ is minimal and for no $x_i, y_i \in X_i : y_i \cgeq_i x_i \cgeq_i z_i$ we have $j \R^{x_iz_{-i}} i$ and NOT $j
    \R^{y_iz_{-i}} i$.
  \end{itemize}
  Define \emph{$\NW{ij}$} as a union of the following three sets:
  \begin{itemize}
  \item All $z \in X$ such that $j \R^{z} i$, if $z_j$ is not maximal and $z_i$ is not minimal;
  \item All $z \in X$ such that $z_i$ is minimal and for no $x_j, y_j \in X_j : y_j \cgeq_j x_j \cgeq_j z_j$ we have $i \R^{x_jz_{-j}} j$ and NOT $i
    \R^{y_jz_{-j}} j$;
  \item All $z \in X$ such that $z_j$ is maximal and for no $x_i, y_i \in X_i : z_i \cgeq_i x_i \cgeq_i y_i$ we have $i \R^{x_iz_{-i}} j$ and NOT $i
    \R^{y_iz_{-i}} j$.
  \end{itemize}
\end{definition}

Presence of maximal and minimal points significantly complicates the definitions of $\SE{ij}$ and $\NW{ij}$, since at such points some of
the sets $\SE[z]{ij}$ and $\NW[z]{ij}$ become degenerate and condition \textbf{3C-ij} trivially holds. If sets $X_i$ and $X_j$ do not
contain minimal or maximal points, we can drop the corresponding conditions in each definition and simply state that
$\SE{ij} = \{z : i \R^z j \}$ and $\NW{ij} = \{z : j \R^z i \}$.

Partial orders $\S_i$ define subsets of the set $X$ as follows.

\begin{definition}
  We write $X^{\S_i} = \bigcap \limits _{(k,j) : k \R_i j} \SE{kj}$
\end{definition}



It is well known that the sufficient property for an additive representation to exist on a Cartesian product is strong independence
\cite{krantz1971foundation}. In the $X = Y^n$ case, the Choquet integral was previously axiomatized using comonotonic strong independence
(or comonotonic trade-off consistency \cite{wakker1991additive}). In this paper we will be using sets $X^{\S_i}$ to formulate a similar
condition.

\begin{definition}
  We say that $i \in N$ is essential on $A \subset X$ if there exist $x_ix_{-i}, y_ix_{-i} \in A$, such that $x_ix_{-i} \cgt y_ix_{-i}$.
\end{definition}
%
%
\begin{description}
  \item[A4 - Intra-coordinate trade-off consistency] 
    \begin{equation*}
      \left.
      \begin{aligned}
        a_{i}x_{-i} & \cleq b_{i}y_{-i} \\ 
        a_iw_{-i}   & \cgeq b_iz_{-i}   \\
        c_ix_{-i}  & \cgeq d_iy_{-i}
      \end{aligned}
    \right\} \Rightarrow
    c_iw_{-i} \cgeq d_iz_{-i},
    \end{equation*}
provided that either:
    \begin{enumerate}[label={\alph*})]
    \item Exists $X^{\S_j}$ such that $ a_{i}x_{-i}, b_{i}y_{-i}, a_iw_{-i}, b_iz_{-i}, c_ix_{-i}, d_iy_{-i}, c_iw_{-i}, d_iz_{-i} \in X^{\S_j}$
    \item Exist $X^{\S_j},X^{\S_k}$ such that $ a_{i}x_{-i}, b_{i}y_{-i}, a_iw_{-i}, b_iz_{-i} \in X^{\S_j} $, $i$ is essential on $X^{\S_j}$, and
      $c_ix_{-i}, d_iy_{-i}, c_iw_{-i}, d_iz_{-i} \in X^{\S_k}$, or;
    \item Exist $X^{\S_j},X^{\S_k}$ such that $ a_{i}x_{-i}, b_{i}y_{-i}, c_ix_{-i}, d_iy_{-i} \in X^{\S_j}$, $i$ is essential on $X^{\S_j}$, and $
      a_iw_{-i}, b_iz_{-i}, c_iw_{-i}, d_iz_{-i} \in X^{\S_k}$.
    \end{enumerate} \label{sec:am3}
  \end{description}

  Informally, the meaning of the axiom is that ordering between preference differences (``intervals'') is preserved irrespective of the
  ``measuring rods'' used to measure them. However, contrary to the additive case this does not hold on all $X$, but only when either points
  involved in all four relations lie in the same ``3C-set'' $X^{S_j}$, or points involved in two relations lie in one such set and those
  involved in the other two in another.
\begin{description}
\item[A5 - Inter-coordinate trade-off consistency]
    \begin{equation*}
      \left.
      \begin{aligned}
        a_ix_{-i} & \cleq b_iy_{-i} \\ 
        c_ix_{-i} & \cgeq d_iy_{-i} \\
        a_iy^0_{-i} & \sim p_jx^0_{-j} \\
        b_iy^0_{-i} & \sim q_jx^0_{-j} \\
        c_iy^1_{-i} & \sim r_jx^1_{-j} \\
        d_iy^1_{-i} & \sim s_jx^1_{-j} \\
        p_je_{-j} & \cgeq q_jf_{-j}
      \end{aligned}
    \right\} \Rightarrow
    r_je_{-j} \cgeq s_jf_{-j} 
    \end{equation*}
    for all $a_ix_{-i}, b_iy_{-i}, c_ix_{-i}, d_iy_{-i} \in X^{\S_j}$ provided $i$ is essential on $X^{\S_j}$, $a_iy^0_{-i}, b_iy^0_{-i},
    c_iy^1_{-i}, d_iy^1_{-i} \in X^{\S_k}$, $p_jx^0_{-j}, q_jx^0_{-j}, r_jx^1_{-j}, s_jx^1_{-j} \in X^{\S_l}$
    provided $j$ is essential on $X^{\S_l}$, $ p_je_{-j}, q_jf_{-j}, r_je_{-j}, s_jf_{-j} \in X^{\S_m}$.
    \label{sec:am5}
  \end{description}
  The formal statement of the \textbf{A5} is rather complicated, but it simply means that the ordering of the ``intervals'' is preserved
  across dimensions. Together with \textbf{A4} the conditions are similar to Wakker's trade-off consistency condition
  \cite{wakker1991additive-RO} 
  . The axiom bears even stronger similarity to Axiom 5 (compatibility) from section 8.2.6 of \cite{krantz1971foundation}. Roughly speaking,
  it says that if the ``interval'' between $c_i$ and $d_i$ is ``larger'' than that between $a_i$ and $b_i$, then ``projecting'' these
  intervals onto another dimension by means of the equivalence relations must leave this order unchanged. We additionally require the
  comparison of intervals and ``projection'' to be consistent - meaning that each quadruple of points in each part of the statement belongs
  to the same $X^{\S_i}$. Another version of this axiom, which is used frequently in proofs, can be formulated in terms of standard
  sequences (Lemma \ref{lm:A5}).
\begin{description}
\item[A6 - Bi-independence] Let $a_ix_{-i},b_ix_{-i},c_ix_{-i},d_ix_{-i} \in X^{\S_i}$ and $a_ix_{-i} \cgt b_ix_{-i}$. If for some
  $y_{-i} \in X_{-i}$ we have $c_iy_{-i} \cgt d_iy_{-i}$, then $c_ix_{-i} \cgt d_ix_{-i}$ for all $i \in N$.
  \end{description}
  This axiom is similar to ``strong monotonicity'' in \cite{wakker1991additive-RO}. We analyze its necessity and the intuition behind it in
  section \ref{sec:essentiality}.
\begin{description}
  \item[A7 - Essentiality] All coordinates are essential on $X$. 
  \item[A8 - Restricted solvability]  If $a_ix_{-i} \cgeq y \cgeq b_ix_{-i}$, then there exists $c: c_ix_{-i} \sim y$ for $i \in N$.
  \item[A9 - Archimedean axiom] Every bounded standard sequence contained in some $X^{\S_i}$ is finite, and in the case of only one
    essential coordinate, there exists a countable order-dense subset of $X^{\S_i}$. 
  \end{description}

Finally, we can introduce a notion of \emph{interacting} coordinates. 
\begin{definition}
  Coordinates $i$ and $j$ are \emph{interacting} if exists $z \in X$, such that $i \S^z j$ or $j \S^z i$. We call a set $A \subset N$ an
  \emph{interaction clique} if for each $i,j \in A$ we can build a chain of coordinates $i, k, \ldots, j$, such that every two subsequent
  coordinates in the chain are interacting. 
\end{definition}

Interaction cliques play an important role in the uniqueness properties of the representation. In what follows we will be considering only
cliques of maximal possible size if not specified otherwise.

\subsection{Additional assumptions}
\label{sec:addd-struct-assumpt}

The following additional assumptions are made. The reasoning behind each one is explained below. They are not required for the construction
of the representation in general.

\begin{description}
\item[``Collapsed'' equivalent points along dimensions.] For no $i \in N$ and no $a_i,b_i \in X_i$ holds $a_ix_{-i} \sim b_ix_{-i}$ for all $x_{-i} \in X_{-i}$. 
\end{description}
If this wasn't true, we could have value functions assigning the same value to several points in the same set $X_i$. To simplify things we
exclude such case, however, it can be easily reconstructed once the representation is built. 
\begin{description}
\item[Density.] We assume that for all $i \in N$, whenever $a_ix_{-i} \cgt b_ix_{-i}$, there exists $c_i \in X_i$ such that
  $a_ix_{-i} \cgt c_ix_{-i} \cgt b_ix_{-i}$ ($X$ is order dense).

\item[``Closedness''.] For every $i$ and $j$, if there exist $x_ix_jz_{-ij}$ such that $ i \S^{x_ix_jz_{-ij}} j$ and
$y_ix_jz_{-ij}$ such that $ j \S^{y_ix_jz_{-ij}} i$, then exists $z_i \in X_i$ such that $i \E^{z_ix_jz_{-ij}} j$.  
\end{description}

This assumption says that sets $\SE{ij}$ and $\NW{ij}$ are ``closed''. In the representation this translates into existence of the inverse
for all points where value functions $f_i$ and $f_j$ are equal, provided $i$ and $j$ are interacting. This is a technical simplifying
assumption and the proof can be done without it.

\begin{description}
\item[Geometry of $X$.] For every clique of interacting variables $A \subset N$, there exist at least two points $r^0_A, r^1_A \in X$ such,
  that for every pair $i,j \in A$, we have $i \E^{r^0_A} j$ and $i \E^{r^1_A} j$.
\end{description}

Again, this is a simplifying assumption, making the proof somewhat less general and closer to the homogeneous case. Without it we can have
a situation, where the smallest value of $f_i:X_i$ is larger then the greatest value of $f_j:X_j$ for some $i,j \in N$. This in turn does
not allow to construct the capacity in a unique way. Another way to stating this assumption, is to say that $X$ must contain points
corresponding to all possible
acyclic partial orders on $N$, generated by interacting pairs $i \S j$. Work to remove this assumption is still in progress.


\section{Representation theorem}
\label{sec:repr-theor}

As follows from the definition of the Choquet integral (Section \ref{sec:choquet-integral}), every point $x \in X$ uniquely corresponds to a set
of weights $p^x_i : p^x_i \geq 0, \sum _{i \in N}p^x_i = 1$. This notation is used to simplify the statement of the following theorems. 

\begin{theorem}
  \label{theo:repr}
  Let $\cgeq$ be an order on $X$ and the structural assumption hold. Then, if axioms \textbf{A1}-\textbf{A9} are
  satisfied, there exists a capacity $\nu$ and value functions $f_1:X_1 \rightarrow \mathbb{R}, \ldots, f_n:X_n \rightarrow \mathbb{R}$,
  such that $\cgeq$ can be represented by the Choquet integral:
  \begin{equation}
    \label{eq:repr}
    x \cgeq y \iff C(\nu,(f_1(x_1),\ldots,f_n(x_n))) \geq C(\nu,(f_1(y_1),\ldots,f_n(y_n))),
  \end{equation}
  for all $x,y \in X$.  
\end{theorem}

Capacity and value functions have the following uniqueness properties. Let ${\cal I} = \{A_1, \ldots, A_k \}$ be a partition of $N$, such
that $m(B) = 0$ for all $B \subset N$ such that $B \cap A_i \neq \emptyset, B \cap A_j \neq \emptyset$. If no such partition exists, let
${\cal I} = \{ N \}$.

\begin{theorem}
  \label{theo:uniqueness}

  Let $g_1:X_1 \rightarrow \mathbb{R}, \ldots, g_n:X_n \rightarrow \mathbb{R}$ be such that (\ref{eq:repr}) holds with $f_i$ substituted by
  $g_i$. Then, at all $x_i \in X_i$, such that for some $z_{-i}$ we have $p^{x_iz_{-i}}_i > 0$, and also $p^{x_iz_{-i}}_j > 0, j \neq i$,
  value functions $f_i$ and $g_i$ are related in the following way:
\begin{equation*}
  f_i(x_i) = \alpha_{A_j}g_i(x_i) + \beta_{A_j},
\end{equation*}

Capacity changes as follows
\begin{equation*}
  m'(B) = \frac{\alpha_{A_j}m(B)}{\sum _{C \subset A_i, A_i \in {\cal I}} \alpha_{A_i}m(C)}.
\end{equation*}

At the remaining points of $X$, i.e. for $x_i$ such that for any $z_{-i} \in X_{-i}$ we have $p^{x_iz_{-i}}_i = 1$, and $p^{x_iz_{-i}}_j =
0$ for all $j \neq i$, value functions $f_i$ have the following uniqueness properties\footnote{Due to our assumption that for no $a_i, b_i$ we have $a_iz_{-i} \sim b_iz_{-i}$ for all $z_{-i}$, we can't
  have $p^{x_iz_{-i}}=0$ for all $z_{-i}$.}:
\begin{equation*}
  f_i(x_i) = \psi_i(g_i(x_i)),
\end{equation*}
where $\psi_i$ is an increasing function, and for all $j \in N, j \neq i$, such that exists $A \in N: i,j \in A, m(A) > 0$, we additionally
have
\begin{equation*}
f_i(x_i) = f_j(x_j) \iff g_i(x_i) = g_j(x_j).
\end{equation*}
.

\end{theorem}



\section{Additive representations on $X^{S_a}$}
\label{sec:addit-repr-xs_a}

We start by removing maximal and minimal elements from the sets $X_i$. The representation will be extended to these points in Section \ref{sec:extend-repr-extr}.

Similar to \citep{wakker1991additive-RO} we will be covering the sets $X^{S_a}$ with ``rectangular'' subsets. Given a point $z \in X^{S_a}$
we construct a ``rectangular'' set $X^{z(S_a)}$ in the following way:
  \begin{itemize}
  \item If $j$ is minimal in $\S_a$, then $X^{z(S_a)}_j = {x_j \in X_j : z_j \cgeq_j x_j}$.
  \item If $j$ is maximal in $\S_a$, then $X^{z(S_a)}_j = {x_j \in X_j : x_j \cgeq_j z_j}$.
  \item If $j$ is neither maximal not minimal, then $X^{z(S_a)}_j = [x_j \in X_j : x_j \cgeq_j z_j, x_jz_{-j} \in X^{S_a}]$. 
  \item If for no $k$ we have $j \S_a k$ or $k \S_a j$, then $X^{z(S_a)}_j = X_j$.
  \end{itemize}

\subsection{Constructing additive representation on $X^{z(S_a)}$}
\label{sec:constr-addit-funct}

We assume that $X^{S_a}$ has at least two essential coordinates. By Lemma \ref{lm:strong-com-monotonicity}, all sets $X^{z(S_a)}$ therefore have at least two
essential coordinates. Moreover, the essential coordinates are the same across all sets. 

\begin{theorem}
  \label{theo:additive-XzSa}
  For any $z \in X^{S_a}$ there exists an additive representation of $\cgeq$ on $X^{z(S_a)}$:
  \begin{equation*}
    x \cgeq y \Leftrightarrow \sum _{i=1} ^n V^z_i(x_i) \geq \sum _{i=1} ^n V^z_i(y_i),
  \end{equation*}
  for all $x,y \in X^{z(S_a)}$.
\end{theorem}

\begin{proof}

  $X^{z(S_a)}$ is a Cartesian product, $\cgeq$ is a weak order on $X^{z(S_a)}$, $\cgeq$ satisfies generalized triple cancellation on $X^{z(S_a)}$,
  $\cgeq$ satisfies Archimedean axiom on $X^{z(S_a)}$, at least two coordinates are essential. It remains to show that $\cgeq$ satisfies restricted
  solvability on $X^{z(S_a)}$.
  
  Assume that for some $x_iz_{-i}, w, y_iz_{-i} \in X^{z(S_a)}$, we have $x_iz_{-i} \cgeq w \cgeq y_iz_{-i}$, hence exists
  $z_i \in X_i: z_iz_{-i} \sim w$. We need to show that $z_iz_{-i} \in X^{S_a}$. If $w \sim x_iz_{-i}$ or $w \sim y_iz_{-i}$, then the
  conclusion is immediate (since either point belongs to $X^{z(S_a)}$). Hence, assume $x_iz_{-i} \cgt z_iz_{-i} \cgt y_iz_{-i}$. This means
  that $x_i \cgeq_i z_i \cgeq_i y_i$. Since $z_i$ is ``sandwiched'' between $x_i$ and $y_i$ we conclude that for any $j \in N \setminus i$,
  $i \S^{x_iz_{-i}} j$ and $i \S^{y_iz_{-i}} j$ imply also $i \S^{z_iz_{-i}} j$, and symmetrically $j \S^{x_iz_{-i}} i$ and
  $j \S^{y_iz_{-i}} i$ imply $j \S^{z_iz_{-i}} i$. Hence, it is also in $X^{S_a}$.

  Therefore all conditions for the existence of an additive representation are met \citep{wakker1991additive}.  
\end{proof}

\subsection{Joint representation $V^{S_a}$ on $X^{S_a}$}
\label{sec:addit-funct-vs_a}

This section is based on \mbox{\citep{wakker1991additive-RO}} with some modifications.

\begin{theorem}
  \label{theo:Vsa-on-Xsa}
  There exists an additive interval scale $V^{S_a}(z) = \sum _{i = 1}^n V^{S_a}_i(z_i)$ on $X^{S_a}$, which represents $\cgeq$ on every
  $X^{z(S_a)}$ with $z \in X^{S_a}$.
\end{theorem}
\begin{proof} 

  Choose the reference set - pick any $r \in X^{S_a}$ such that $X^{r(S_a)}_i$ contains more than two points for any $X^{S_a}$-essential
  $i$. Choose a ``zero'' point - any $r^0 \in X^{r(S_a)}$, and a ``unit mark'' - a point $r_k^1r^0_{-k} \in X^{r(S_a)}$, such that:
  \begin{itemize}
  \item $k$ is essential on $X^{S_a}$,
  \item $r^1_k \cgeq_k r^0_k$.
  \end{itemize}
  Set $V^{r}_i(r^0_i) = 0$ for all $i \in N$ and $V^{r}_k(r^1_k) = 1$.  This uniquely defines unit and locations of all $V^r_i, i \in N$.  

In the following we assume that sets $X^{z(S_a)}_i, X^{z(S_a)}_k$ each contain at least two points, otherwise, alignment is trivial.

Assume $X^{r(S_a)}_i \cap X^{z(S_a)}_i = \emptyset$ and $X^{r(S_a)}_k \cap X^{z(S_a)}_k = \emptyset$ (variations are all covered by the
below procedure). We will construct two auxiliary points $z'$ and $r'$ such that $X^{z'(S_a)}_i \subset X^{z(S_a)}_i$, $X^{z'(S_a)}_i \subset
X^{r(S_a)}_i$, $X^{r'(S_a)}_k \subset X^{z(S_a)}_k$, $X^{r'(S_a)}_k \subset X^{r(S_a)}_k$. It would allow us to align first $V^r$ and
$V^{r'}$, then $V^{r'}$ and $V^{z'}$, and finally $V^{z'}$ and $V^z$. 

Construct the point $z'$ by taking coordinate-wise maxima of $r$ and $z$ for coordinates $j$ such that $j \R_a i$, not including $i$
itself, and coordinate-wise minima of $r$ and $z$ for coordinates $j$, such that $i \R_a j$ and $i$ itself. In the short notation the first
point is $z' := \max(r_j,z_j)_{j:j \R_a i} \min(r_j,z_j)_{j=i, j:i \R_a j}$. The second point $r'$ is constructed by taking coordinate-wise maxima of $r$ and $z$ for
coordinates $j$ such that $j \R_a k$, not including $k$ itself, and coordinate-wise minima of $r$ and $z$ for coordinates $j$, such that
$k \R_a j$ and $k$ itself. In the short notation the second point looks like $r' := \max(r_j,z_j)_{j: j \R_a k} \min(r_j,z_j)_{j=k, j: k \R_a j}$.

Note that both points are in $X^{S_a}$ since relations $j \R_a l$ remain intact for all pairs $j,l$. Note also, that $X^{z'(S_a)}_i$
contains both $X^{z(S_a)}_i$ and $X^{r(S_a)}_i$, and $X^{r'(S_a)}_k$ contains both $X^{z(S_a)}_k$ and $X^{r(S_a)}_k$.

Now we have that sets $(X^{z(S_a)}_i \times X^{z(S_a)}_k) \cap (X^{z'(S_a)}_i \times X^{z'(S_a)}_k)$, $(X^{z'(S_a)}_i \times X^{z'(S_a)}_k)
\cap (X^{r'(S_a)}_i \times X^{r'(S_a)}_k)$, $(X^{r'(S_a)}_i \times X^{r'(S_a)}_k) \cap (X^{r(S_a)}_i \times X^{r(S_a)}_k)$ are all
non-empty, and each dimension contains more than two points. Relation $\cgeq_{i,k}$ on these sets satisfies Archimedean axiom, restricted
solvability, and \textbf{A4}. Hence we can apply standard uniqueness properties of additive representations. We first align $V^{r'}_k$ with
$V^r_k$ and $V^{r'}_i$ with $V^r_i$, then $V^{z'}_k$ with $V^{r'}_k$ and $V^{z'}_i$ with $V^{r'}_i$, and finally $V^{z}_k$ with $V^{z'}_k$
and $V^{z}_i$ with $V^{z'}_i$ by changing the common unit and locations of corresponding value functions. 

Having aligned like this $V^z_i$ and $V^z_k$ with $V^r_i$ and $V^r_k$ for all $z \in X^{S_a}$ we can perform the same alignment operation
for all remaining essential coordinates $j$, using pairs $V^z_j$ and $V^z_k$. At this stage, functions $V^z_k$ are already aligned, hence
have a correct unit and location. As above, uniqueness properties of additive representations of relation $\cgeq_{j.k}$ imply that the unit
of functions $V^z_j$ is already aligned with that of $V^r_j$ and only location change has to be performed. This can also be done as above. 

Once such alignment has been performed for all essential coordinates, we can verify that this is done consistently throughout $X^{S_a}$. In
particular, for any $s$ and $t$ from $X^{S_a}$ we must be able to show that for any essential $j \in N$, we have $V^s_j = V^t_j$ on
$X^{s(S_a)}_j \cap X^{t(S_a)}_j$. To show this a following argument can be used. During the initial alignment of $V^s_j$ and $V^t_j$,
auxiliary points $t'$ and $s'$ were used, such that $X^{s'(S_a)}_j$ includes $X^{s(S_a)}_j$ and $X^{r(S_a)}_j$, and $X^{t'(S_a)}_j$ includes
$X^{t(S_a)}_j$ and $X^{r(S_a)}_j$. Hence, functions $V^{s'}_j$ and $V^{t'}_j$ coincide with $V^{r}_j$ on $X^{r(S_a)}_j$. To show that they
coincide on all common domain, including $X^{s(S_a)}_j \cap X^{t(S_a)}_j$, we just need to follow the same procedure as before and construct
a point that contains $X^{s'(X_a)}_k$ and $X^{t'(X_a)}_k$ for some essential $k$. Then a uniqueness argument can be evoked once again, and
since $V^{s'}_j$ and $V^{t'}_j$ coincide on $X^{r(S_a)}_j$, they would necessarily coincide also on the remaining common domain, which
includes $X^{s(S_a)}_j \cap X^{t(S_a)}_j$. Finally, since $V^{s}_j = V^{s'}_j$ on $X^{s(S_a)}_j$, and $V^{t}_j = V^{t'}_j$ on on
$X^{t(S_a)}_j$, we get that $V^{s}_j = V^{t}_j$ on $X^{s(S_a)}_j \cap X^{t(S_a)}_j$.

At this point we can drop the superscripts and define functions $V^{S_a}_i$ which coincide with $V^{z(S_a)}_i$ for all $z \in X^{S_a}$ on
the corresponding domains. By the above argument, these functions are well-defined.

\end{proof}

\subsection{$V^{S_a}$ is globally representing on $X^{S_a}$}
\label{sec:vs_a-glob-repr}

\begin{lemma}
  \label{lm:2}
  For all $X^{S_a}$-essential $i \in N$, $V^{S_a}_i$ represents $\cgeq_i$ on $X^{S_a}_i$. 
\end{lemma}
\begin{proof}
  Let $\alpha_i,\beta_i \in X^{S_a}_i$ be such that $\alpha_i \cgeq_i \beta_i$. Similarly to the construction of $r'$ and $z'$ in the proof
  of theorem \ref{theo:Vsa-on-Xsa}, we can show that always exists $x_i$ such that $\alpha_ix_{-i}, \beta_ix_{-i} \in X^{S_a}$. The
  conclusion follows.
\end{proof}

\begin{theorem}
  \label{theo:Vsa-global}
  Representation $V^{S_a}$ obtained in Theorem \ref{theo:Vsa-on-Xsa} is globally representing on $X^{S_a}$.
\end{theorem}

\begin{proof}
  We need to show that  $x \cgeq y \iff V^{S_a}(x) \geq V^{S_a}(y)$.
  \begin{itemize}
  \item If exists $z$ such that $x,y \in X^{z(S_a)}$ then the result is immediate.
  \item If the above is not true, we will show that exists $x' \sim x$ such that $V^{S_a}(x) = V^{S_a}(x')$ and $x'_i \cgeq_i y_i$ for all $i$.
  \end{itemize}

  The procedure is identical to \cite{wakker1991additive} with some minor modifications.
  \begin{enumerate}
  \item Find $i$ such that $y_i \cgt_i x_i$ and $x_k \cgeq_k y_k$ for all $k$ such that $k \S_a i$. We have $y_ix_{-i} \in X^{S_a}$ (since
    for all $k \in N$ such that $k \S_a i$ we have $x_k \cgeq y_k$, hence $k \R^{y_iy_{-i}} i$ implies $k \R^{y_ix_{-i}} i$, whereas for all
    $t \in N$ such that  $i \S_a t$ we have $i \R^{x_ix_{-i}} t$, hence $i \R^{y_ix_{-i}}$).
  \item Similarly, find $j$ such that $x_j \cgt_j y_j$ and $y_k \cgeq_k x_k$ for all $k$ such that $j \S_a k$. By similar reasoning, $y_jx_{-j} \in X^{S_a}$.
  \item We are increasing $x_i$ and decreasing $x_j$ and thus move in the direction of $y$. 
  \item Note, that $x_{-ij}y_iy_j \in X^{S_a}$.
  \item If $x_{-ij}y_iy_j \cgeq x$, then by restricted solvability ($ x_{-ij}y_iy_j \cgeq x \cgeq x_{-ij}x_iy_j$) exists $x' :=
    x_{-ij}x'_iy_j \sim x$, where $y_i \cgeq_i x'_i \cgeq_i x_i$. If $x \cgt x_{-ij}y_iy_j$, then by restricted solvability ($x_{-ij}y_ix_j
    \cgeq x \cgt x_{-ij}y_iy_j $) exists $x' := x_{-ij}y_ix'_j \sim x$, and $x_j \cgeq_j x'_j \cgeq_J y_j$.
  
  \item In both cases, the resulting point $x'$ is in $X^{S_a}$, moreover $x', x \in X^{z(S_a)}$ where $z := x_{ij}x_iy_j$, hence $x'$ has
    the same $V^{S_a}$-value as $x$, but one more coordinate becomes identical to that of $y$.
 
  \item After repeating the procedure unless $x'_i \cgeq_i y_i$ (at most $n$ times), we get the result by Lemma \ref{lm:2}.
  \item Moreover, if $x \sim y$, we at the end of the procedure we would necessarily arrive to $y$ itself (by strong monotonicity as in Lemma
    \ref{lm:strong-com-monotonicity}, and structural assumption \textbf{SA1}). Hence we get $x \cgt y \Rightarrow V^{S_a}(x) > V^{S_a}(y)$ and $x \sim y
    \Rightarrow V^{S_a}(x) = V^{S_a}(y)$, which implies that $x \cgeq y \iff V^{S_a}(x) \geq V^{S_a}(y)$
  \end{enumerate}
\end{proof}


\section{Aligning cardinal representations for different $X^{S_a}$}
\label{sec:align-addit}

There can be several cases depending on what variables are essential on various sets $X^{S_i}$.  We start with the case where exist
$X^{S_a}$ and $X^{S_b}$ having at least two essential variables each.

\subsection{Exist at least two sets $X^{S_i}$ with at least two essential coordinates}
\label{sec:both-coordinates-are}

%

\begin{theorem}
  \label{theo:Vi-proportional}
  Assume that at least two coordinates are essential on $X^{S_a}$ and $X^{S_b}$. For any $i \in N$ that is essential on both areas, it holds
  $V^{S_a}_i(z_i) = \lambda^{ab}_iV^{S_b}_i(z_i)$ for all $z_i$ from the common domain of $V^{S_a}_i(z_i)$ and $V^{S_b}_i(z_i)$, if a common
  location is chosen for both functions.
\end{theorem}
\begin{proof}
%
%
  If the common domain of $V^{S_a}_i(z_i)$ and $V^{S_b}_i(z_i)$ is empty or contains just one point, the result is
  trivial. \footnote{Obviously, the common domain is not empty if we assume \textbf{SA4}, in which case $r^0_i$ and $r^1_i$ are in the common domain
    by assumption}  Assume that
  $i,j$ are essential on $X^{S_a}$ and $i,l$ are essential on $X^{S_b}$. First, we will establish that a standard sequence on coordinate $i$
  in $X^{S_a}$ is also a standard sequence in $X^{S_b}$ (provided all points of the sequence lie within a common domain of $V^{S_a}_i(z_i)$
  and $V^{S_b}_i(z_i)$). This follows from \textbf{A4}. Build any standard sequence $X^{S_a}_i$, say
  $\{\alpha^{k}_i : \alpha^k_iv_jx_{-ij} \sim \alpha^{k+1}_iw_jx_{-ij} \}$. Then,
  $\{\alpha^{k}_i : \alpha^k_it_lx_{-il} \sim \alpha^{k+1}_iu_lx_{-il}\}$ is a standard sequence in $S_b$, i.e. if exist $t_l,u_l \in X_l$
  such that $\alpha^k_it_lx_{-il} \sim \alpha^{k+1}_iu_lx_{-il}$ for some $k$, then by \textbf{A4}:
  \begin{equation*}
\left.
  \begin{aligned}
    \alpha^k_iv_jx_{-ij} & \sim \alpha^{k+1}_iw_jx_{-ij}\\
    \alpha^k_it_lx_{-il} & \sim \alpha^{k+1}_iu_lx_{-il}\\
    \alpha^{k+1}_iv_jx_{-ij} & \sim \alpha^{k+2}_iw_jx_{-ij}
  \end{aligned}
\right\}
\Rightarrow     \alpha^{k+1}_it_lx_{-il} \sim \alpha^{k+2}_iu_lx_{-il}
  \end{equation*}

  Pick two points $r^0_i$ and $r^1_i$ in the common domain and set $V^{S_a}_i(r^0_i) = V^{S_b}_i(r^0_i) = 0$. Assume we now have $V^{S_a}_i(r^1_i) =
v_a$ and $V^{S_b}_i(r^1_i) = v_b$. We need to show that for any point $z_i$ from the common domain of $V^{S_a}_i$ and $V^{S_b}_i$ we have
$V^{S_a}_i(z_i) = \lambda^{ab}_iV^{S_a}_i(z_i)$, where $\lambda^{ab}_i = \frac{v_b}{v_a}$.

  Build standard sequences from $r^0_i$ to $r^1_i$, and from $r^0_i$ to $z_i$. We have
  \begin{equation*}
    \begin{aligned}
      V^{S_a}_i(r^1_i) - V^{S_a}_i(r^0_i) & \approx  n[V^{S_a}_j(v_j) - V^{S_a}_j(w_j)] \\
      V^{S_a}_i(z_i) - V^{S_a}_i(r^0_i) & \approx  m[V^{S_a}_j(v_j) - V^{S_a}_j(w_j)].
    \end{aligned}
  \end{equation*}
  $V^{S_a}_i(r^0_i) = 0$, hence 
  \begin{equation*}
    V^{S_a}_i(z_i) \approx \frac{mV^{S_a}_i(r^1_i)}{n}.
  \end{equation*}
  Such $n$ and $m$ exist by the Archimedean axiom.
  By the argument above we get 
    \begin{equation*}
    \begin{aligned}
      V^{S_b}_i(r^1_i) - V^{S_b}_i(r^0_i) & \approx  n[V^{S_b}_j(t_l) - V^{S_b}_j(u_l)] \\
      V^{S_b}_i(z_i) - V^{S_b}_i(r^0_i) & \approx  m[V^{S_b}_j(t_l) - V^{S_b}_j(u_l)].
    \end{aligned}
  \end{equation*}
Similarly,
  \begin{equation*}
    V^{S_b}_i(z_i) \approx \frac{mV^{S_b}_i(r^1_i)}{n}.
  \end{equation*}
  By density, we can pick an arbitrary small step of the standard sequences, so the ratio $\frac{m}{n}$ converges to a limit. Thus, finally
  \begin{equation*}
    \begin{aligned}
      V^{S_a}_i(z_i) = \frac{V^{S_a}_i(r^1_i)}{V^{S_b}_i(r^1_i)}V^{S_b}_i(z_i) = \frac{v_a}{v_b}V^{S_b}_i(z_i) = \alpha^{ab}_iV^{S_b}_i(z_i).
    \end{aligned}
  \end{equation*}
\end{proof}

We proceed by picking common locations for all value functions. Since $r^0$ belongs to all $X^{S_a}$, we can set $V^a_i(r^0_i) = 0$. At this
point we can drop superscripts and say that we have representations $\lambda^a_iV_i + \ldots + \lambda^a_nV_n$ on each $X^{S_a}$, defining
also $\lambda^a_i := 0$ for variables $i$ that are inessential on the set $X^{S_a}$.

\subsection{Final rescaling}
\label{sec:final-rescaling}

By assumption, we have two points two points $r^0$ and $r^1$ such that $i \E^{r^0} j$ and $i \E^{r^1} j$ for every interacting $i,j$. From
this follows, that both points belong to every $X^{S_i}$. We can assume that $r^1_i \cgeq_i r^0_i$ for all
$i \in N$ (for variables not interacting with others we can take it to be so, for others see results in Section \ref{sec:shape-z_ij:-i}). Set $V_i(r^0_i) = 0$ for all $i \in N$. Choose some $j \in N$, such that $j$ is essential on at least one $X^{S_a}$, which has
two or more essential variables (including $j$). Set $V_j(r^1_j) = 1$. This sets unit and location for all functions $V_i$ such, that $i$ is essential on some $X^{S_a}$ where
at least one more coordinate is essential. For each such $V_i$, we now have $V_i(r^1_i) = k_i$ (thus $k_j = 1$). Define
$\phi_i := \frac{V_i}{k_i}$, for all $i \in N$.  Additive representations on various $X^{S_a}$ now have the form
$\lambda^a_1k_1\phi_1(x_1) + \ldots + \lambda^a_nk_n\phi_n(x_n)$. Finally, re-scale one more time by dividing everything by the sum of
coefficients:
\begin{equation*}
  \frac{\lambda^a_1k_1}{\sum_{i=1}^n \lambda^a_ik_i}\phi_1(x_1) + \ldots + \frac{\lambda^a_nk_n}{\sum_{i=1}^n \lambda^a_ik_i}\phi_n(x_n)
\end{equation*}
Denoting $\alpha^a_j = \frac{\lambda^a_jk_j}{\sum_{i=1}^n \lambda^a_ik_i}$, we arrive to:
\begin{equation*}
  \phi^a(x) := \alpha^a_1\phi_1(x_1) + \ldots + \alpha^a_n\phi_n(x_n),
\end{equation*}
note that $\sum_{i=1}^n \alpha_i = 1$.

Note that here we set $\phi_i(r^0_i) = 0 $ and $\phi_i(r^1_i) = 1$ for all $i$. As will be shown in the Section
\ref{sec:uniqueness}, this can be relaxed - origin and scaling factors can be chosen individually for each clique.



\section{Constructing global representation on $X$}
\label{sec:global-repr-on-x}

At this stage we can show that representations $\phi^a(x) = \alpha^a_1\phi_1(x_1) + \ldots + \alpha^a_n\phi_n(x_n)$ assign the same value to
equivalence classes of $\cgeq$ in all $X^{S_i}$. To simplify the construction in the main theorem of this section, we introduce the
following lemma.

\begin{lemma}
\label{lm:eqiv-domin}
For every $X^{S_a}$ and every $z \in X^{S_a}$, such that $r^0 \cgt z$, we can find $z'$, such that $z' \sim z$, $z' \in X^{S_a}$, and $r^0_i
\cgeq_i z_i$. Likewise, for every $y \in X^{S_a}$, such that $y \cgt r^0$, we can find $y'$, such that $y' \sim y$, $y' \in X^{S_a}$, and
$y_i \cgeq_i r^0_i$. 
\end{lemma}

\begin{proof}
We can use the same procedure as was used in the proof of Theorem \ref{theo:Vsa-global}. For the case $y \cgt r^0$ the procedure is exactly
the same, while for $r^0 \cgt z$ it is symmetric, as we are moving $z$ this time, and not $r^0$. 
%
%
\end{proof}

Notice that as a result of the rescaling made in Section \ref{sec:final-rescaling}, points $r^0$ and $r^1$ have the same values (0 and 1) in all
$X^{S_a}$ (since values of all $\phi_i$ are equal, weights $\alpha^a_i$ sum up to 1, and weights of all inessential variables are zero). 

\begin{theorem}
\label{theo:eq-util-eqiv-class}
  Let each of $X^{S_a}$ and $X^{S_b}$ have at least two essential variables. Then for any $x \in X^{S_a}, y \in X^{S_b}$ we have $x \cgeq y$
  iff $\phi^a(x) \cgeq \phi^b(y)$. 
\end{theorem}

\begin{proof}
%
%
  First take $x \sim y$, such that $x \in X^{S_a}$, and $y \in X^{S_b}$. If $x \sim y \sim r^0$ or $x \sim y \sim r^1$, the
  conclusion is immediate, so assume otherwise. Let $x \cgt r^0$. Using Lemma \ref{lm:eqiv-domin} we construct $x' \in X^{S_a}$ and $y' \in
  X^{S_b}$, such that $x \sim y \sim x' \sim y'$ and $x'_i \cgeq_i r^0_i$, while $y'_i \cgeq_i r^0_i$.

  Next, build equispaced sequences from $r^0$ to $r^1$ in $X^{S_a}$ and $X^{S_b}$, such that first steps of each sequence are equivalent
  (see details in Section \ref{sec:equispaced-sequences}). By \textbf{A5} the number of steps in both sequences is equal.

  Finally, build sequences from $r^0$ to $x'$ and $y'$ (coordinate-wise dominance simplifies construction of the sequences). The number of
  steps again must be equal, hence the ratios between the number of steps it takes to reach $r^1$ and $x'$, and between the number of steps
  it takes to reach $r^1$ and $y'$ are equal, and hence taking the limit, we get $\phi^a(x) = \phi^a(x') = \phi^b(y') = \phi^b(y)$.

  The same approach applies for $x \cgt y$. By \textbf{A5} the number of steps in the equispaced sequence from $r^0$ to $x$ must be greater,
  than in the sequence from $r^0$ to $y$. Hence also $\phi^a(x) > \phi^b(x)$.

  We now have $x \sim y \Rightarrow \phi^a(x) = \phi^b(y)$ and $x \cgt y \Rightarrow \phi^a(x) > \phi^b(y)$. This implies that $x \cgeq y
  \iff \phi^a(x) \geq \phi^b(y)$.
\end{proof}

At this point we can define value functions on the areas with a single essential variable. Theorem \ref{theo:eq-util-eqiv-class} establishes
that all areas with two or more essential coordinates assign the same value to points from the same equivalence class. Define the value
assigned to an equivalence class belonging to an area with a single essential coordinate to be the same, as it has in some area with two or
more essential coordinates. Such equivalence classes must exist (e.g. those containing points $r^0$ and $r^1$). Finally, define
$\alpha^a_i = 1$ for the essential coordinate $i$ and $\alpha^a_j = 0$ for all other $j$. If after this procedure there remain points in
some $X_i$, for which $\phi_i$ is not yet defined, then we have some equivalence classes to which none of the representations $\phi^k$
assign any value. Since all equivalence classes found in the sets $X^{S_k}$, which have two or more essential variables, by now have a
defined value, such classes are entirely within sets, that have only a single essential variable. Hence, we can trivially extend the
representations, and get also $\phi_i(x_i) > \phi_i(y_i)$ iff $x_i \cgeq_i y_i$ (see Lemma \ref{lm:single-essential-equivalence} which shows
that functions are well-defined).

\begin{lemma}
  \label{lm:A-NA}
  Given $\S_a$ and $\S_b$, such that exists $A \subset N$, for which we have $i \R_a j$ iff $i \R_b j$ for all
  $i \in A, j \in N \setminus A$, the following is true
  \begin{equation*}
    \sum _{i \in A} \alpha^{a}_i = \sum _{i \in A} \alpha^{b}_i.
  \end{equation*}
\end{lemma}

\begin{proof}
  Consider $r^1_Ar^0_{-A}$, which belongs both to $X^{S_a}$ and $X^{S_b}$. By Theorem \ref{theo:eq-util-eqiv-class} and the above
  construction of value functions for the areas with a single essential coordinate, we have $\phi^a(r^1_Ar^0_{-A}) = \phi^b(r^1_Ar^0_{-A})$,
  hence $\sum _{i \in A} \alpha^{a}_i \phi(r^1_i) = \sum _{i \in A} \alpha^{b}_i \phi(r^1_i)$, from which the conclusion follows as
  $\phi_i(r^1_i) = 1$ for all $i \in N$.
\end{proof}

Now we can proceed with construction of a unique \emph{capacity} $\nu: 2^N \rightarrow \mathbb{R}$ from coefficients $\alpha^a_i$ which
exist on various $X^{S_a}$. As shown in \cite{wakker1989additive}, the condition of Lemma \ref{lm:A-NA} is a necessary requirement for
this. Capacity $\nu$ also has a unique M{\"o}bius transform $m: 2^N \rightarrow \mathbb{R}$ (see definition in Section \ref{sec:model}).

We can now construct a representation very similar to the Choquet integral. In order to so, let us define the following function first:
$\Phi_{\wedge}(x, A) := \phi_i(x_i)$ for $i$ such that $j \R^x i$ for all $j \in A \setminus i$, in case when this is true for several $i$,
any can be chosen. We can now construct a global value function (cf. Section \ref{sec:model}):
\begin{equation}
  \label{eq:mobuis-proto-choquet}
  \phi(x) := \sum _{A \in N} m(A) \Phi_{\wedge}(x,A).
\end{equation}
It is easy to see that for each $x \in X$ and every $X^{S_i}$, such that $x \in X^{S_i}$, we have $\phi^a(x) = \phi(x)$.  Now we can show
that $\phi_i(x_i) = \phi_j(x_j)$ whenever $i \E^x j$, providing $i,j$ interact.

\begin{lemma}
  \label{lm:theta-ij-eq}
  For any non-extreme $x \in X$ it holds:
  \begin{equation*}
    i \E^x j \Rightarrow \phi_i(x_i) = \phi_j(x_j),
  \end{equation*}
unless $i$ and $j$ do not interact.
\end{lemma}

\begin{proof}
  Assume $x \in X^{S_a}, x \in X^{S_b}$ such that $k S_a l$ whenever $k S_b l$ for all $k,l \in N$ apart from $i,j$, for which we have
  $i S_a j $ and $j S_b i$. By Theorem \ref{theo:eq-util-eqiv-class}, $\phi^a(x) = \phi^b(x)$ and by Lemma \ref{lm:A-NA}, it is trivial to
  show that $\alpha^a_k = \alpha^b_k$ for all $k \neq i,j$, and $\alpha^a_i + \alpha^a_j = \alpha^b_i + \alpha^b_j$.

  We have $\alpha^a_i\phi_i(x_i) + \alpha^a_j\phi_j(x_j) = \alpha^b_i\phi_i(x_i) + \alpha^b_j\phi_j(x_j)$ (other sum components cancel
  out). Dividing by $\alpha^a_i + \alpha^a_j = \alpha^b_i + \alpha^b_j$, we get a convex combination of $\phi_i(x_i)$ and $\phi_j(x_j)$ on
  both sides. From this follows that either $\phi_i(x_i) = \phi_j(x_j)$ or $\alpha^a_i = \alpha^b_i$ and $\alpha^a_j = \alpha^b_j$. 

  Assume the latter. Repeating this operation for all possible combinations of $X^{S_k}$ and $X^{S_l}$ would lead us to the conclusion that
  $m(B) = 0$ for all $B \supset \{i,j\}$, as weights $\alpha^k_i, \alpha^l_i, \alpha^k_j, \alpha^l_j$ do not change when we move from
  $X^{S_k}$ to $X^{S_l}$, and, accordingly, from $\phi^k$ to $\phi^l$. The conclusion results from equation (\ref{eq:mobuis-proto-choquet}).

  Finally, we can show that this implies that $i$ and $j$ do not interact. This means that $ij$-triple cancellation -
    \begin{equation*}
      \left.
        \begin{aligned}
          a_ip_jz_{-ij} & \cleq b_iq_jz_{-ij}\\
          a_ir_jz_{-ij} & \cgeq b_is_jz_{-ij}\\
          c_ip_jz_{-ij} & \cgeq d_iq_jz_{-ij} 
        \end{aligned}
      \right\} \Rightarrow c_ir_jz_{-ij} \cgeq d_is_jz_{-ij},
    \end{equation*}
holds for all $a_i,b_i, c_i, d_i \in X_i$, $p_j,q_j,r_j,s_j \in X_j$, and all $z_{-ij} \in X_{-ij}$. 
To show this, use equation (\ref{eq:mobuis-proto-choquet}) to write the values for all involved points, grouping the sum components as
follows. For example, for $a_ip_jz_{-ij}$:
\begin{equation*}
  \begin{aligned}
    \phi(a_ip_jz_{-ij}) = & \sum _{\substack{A \supset i \\ A \not \supset j}} m(A) \Phi_{\wedge}(a_ip_jz_{-ij},A) + \sum _{\substack{A
        \supset j \\ A \not \supset i}} m(A) \Phi_{\wedge}(a_ip_jz_{-ij},A) + \\ 
    + & \sum _{A \not \supset i,j} m(A) \Phi_{\wedge}(a_ip_jz_{-ij},A) + \sum _{A \supset i,j} m(A) \Phi_{\wedge}(a_ip_jz_{-ij},A).
  \end{aligned}
\end{equation*}
Notice, that due to the above argument, we have $\sum _{A \supset i,j} m(A) \Phi_{\wedge}(a_ip_jz_{-ij},A) = 0$. Also notice that $\sum
_{\substack{A \supset j \\ A \not \supset i}} m(A) \Phi_{\wedge}(a_ip_jz_{-ij},A)$ does not depend on $a_i$, and $\sum _{A \not \supset i,j}
m(A) \Phi_{\wedge}(a_ip_jz_{-ij},A)$ does not depend on $a_i$ or $p_j$. 

Assume, towards a contradiction, that $d_is_jz_{-ij} \cgt c_ir_jz_{-ij}$. Writing values for all points, and summing the first and the third, and
the second and the fourth inequalities gives: 
\begin{equation*}
  \begin{aligned}
    \sum _{\substack{A \supset i \\ A \not \supset j}} m(A) [\Phi_{\wedge}(a_ip_jz_{-ij},A) + \Phi_{\wedge}(d_ip_jz_{-ij},A)] & \leq \sum
    _{\substack{A \supset i \\ A \not \supset j}} m(A) [\Phi_{\wedge}(b_ip_jz_{-ij},A) + \Phi_{\wedge}(c_ip_jz_{-ij},A)] \\
\sum _{\substack{A \supset i \\ A \not \supset j}} m(A) [\Phi_{\wedge}(a_ip_jz_{-ij},A) + \Phi_{\wedge}(d_ip_jz_{-ij},A)] & > \sum
    _{\substack{A \supset i \\ A \not \supset j}} m(A) [\Phi_{\wedge}(b_ip_jz_{-ij},A) + \Phi_{\wedge}(c_ip_jz_{-ij},A)],
  \end{aligned}
\end{equation*}
which is a contradiction. Hence, $ij$-triple cancellation holds for all $a_i,b_i, c_i, d_i \in X_i$, $p_j,q_j,r_j,s_j \in X_j$, and all
$z_{-ij} \in X_{-ij}$, and thus $i$ and $j$ do not interact.
\end{proof}

\begin{lemma}
  If for some $z \in X$ we have $i \S^z j$, then $\phi_i(z_i) > \phi_j(z_j)$. 
\end{lemma}
\begin{proof}
  It is easy to verify that (due to the ``Closedness'' assumption), there exists $x_{ij}z_{-ij}$, such that $i \E^{x_{ij}z_{-ij}} j$ and
  $z_i \cgeq_i x_i$, whereas $x_j \cgeq_j z_j$. Since $\phi_i$ represents $\cgeq_i$, by Lemma \ref{lm:theta-ij-eq}, and the fact that
  $\cgeq_i$ is asymmetric (due to the structural assumption), we have $\phi_i(z_i) > \phi_i(x_i) = \phi_j(x_j) > \phi_j(z_j)$.
\end{proof}

Now we have that $[i \E^x j] \Rightarrow [\phi_i(x_i) = \phi_j(x_j)]$ for interacting $i,j$, and $[i \S^x j] \Rightarrow [\phi_i(x_i) > \phi_j(x_j)]$, hence we
conclude that $[i \R^x j] \iff [\phi_i(x_i) \cgeq \phi_j(x_j)]$ for all interacting $i,j$, which allows us to finally rewrite
(\ref{eq:mobuis-proto-choquet}) as the Choquet integral:

\begin{equation*}
  \phi(x) = \sum _{A \subset N} m(A) \min _{i \in A} f_i(x_i).
\end{equation*}

To summarize the results of this section we state the following lemma:

\begin{lemma}
  \label{lm:representation-without-extreme}
Let $X'_i := X \setminus \{\text{ maximal and minimal elements of } X\}$. Let $X' := X'_1 \times \ldots \times X'_n$. Assume that at least
one of the sets $X^{'S_a}$, defined as previously, has more than two essential variables. Then for every $x,y \in X'$ we have 
\begin{equation*}
  x \cgeq y \iff \phi(x) \geq \phi(y).
\end{equation*}
\end{lemma}

\subsection{Case with a single essential variable on every $X^{S_a}$}
\label{sec:case-with-single}

For this case  we only need \textbf{A3} to construct the representation. Since $\cgeq$ is a weak order and each $X^{S_i}$ has a countable order-dense subset, there exists a function $F:X \Rightarrow \mathbb{R}$,
such that $x \cgeq y \iff F(x) \geq F(y)$. To perform the construction of the value functions we need the following lemma.

\begin{lemma}
\label{lm:single-essential-equivalence}
  Let $x_iz_{-i} \in X^{S_a}$ and $y_iz_{-i} \in X^{S_b}$. Let also $i$ be the only essential coordinate on $X^{S_a}$ and $X^{S_b}$. Then,
  $x_iz_{-i} \sim y_iz_{-i}$. 
\end{lemma}
\begin{proof}
  The idea of the proof is to ``trace a path'' from $X^{S_a}$ to $X^{S_b}$ by constructing a sequence of points and subsets $X^{S_j}$. We
  will keep $x_i$ unchanged, but will move the remaining coordinates, in order to show that each point in the sequence belongs to the
  current and the subsequent subsets, moreover $i$ will be the only essential coordinate on all $X^{S_j}$.

  There are several steps. We start with pairs of coordinates $j,k$ such that $i \S_a j,k$, and $j \S_a k$ but $k \S_b j$ (or vice
  versa). Note that both $j,k$ are inessential in $X^{S_a}$. Let $j \S_a k$, and $j,k$ be subsequent in $\S_a$, i.e. there is no $m$ such
  that $j \S_a m \S_a k$. Such a pair must obviously exist. Using a simple argument, we can show that we can construct a point $x_iz^1_{-i}$
  by changing $z_j$ and $z_k$, such that $k \E^{x_iz'_{-i}} j$. Call this new ordering $\S_1$. By a density argument, we can show that if
  both $j$ and $k$ were not essential on $X^{S_a}$, they will not be essential on $X^{S_1}$ either. Essentiality of all other coordinates is
  also not affected, so $i$ remains the only essential coordinate on $X^{S_1}$. Note, that because $k \E^{x_iz'_{-i}} j$, we have
  $x_iz^1_{-i} \in X^{S_a}$, $x_iz^1_{-i} \in X^{S_1}$. Proceeding like this, we build a sequence $x_iz^j_{-i}$ until eventually all
  pairs $j,k$ such that $i \S_a j$, $i \S_a k$ are ranked in the same order as in $\S_b$. Using the same approach, we repeat it for $j,k$
  such that $j \S_a i$, $k \S_a i$. Eventually, the sequence ends with an order $\S_n$.

It remains to switch the order of coordinates which change the relative position with $i$ in $\S_a$ and $\S_b$. By using the closedness structural
assumption, we can construct a point which will belongs both to $X^{S_n}$ and $X^{S_b}$, and since $i$ is the only essential coordinate on
both sets, we obtain the result.
\end{proof}

Using this lemma, we can now define value functions $\phi_i(x_i) = F(x)$ by picking $x \in X^{S_a}$ where $i$ is the essential
coordinate. Lemma shows that the functions are well-defined.  It remains to construct a capacity. We can do so, by letting $\alpha^a_i = 1$
for essential $i$ and $\alpha^a_j = 0$ for the remaining coordinates. Points $r^0$ and $r^1$ can be used to show a result similar to Lemma
\ref{lm:A-NA}. This means, that there exist a capacity $\nu$, and the preference relation can be represented by a Choquet integral with
respect to this capacity and value functions, defined as above.

An alternative construction is given in Section \ref{sec:altern-treatm-case}.


\section{Extending the representation to the extreme points of $X$}
\label{sec:extend-repr-extr}

\subsection{Definition value functions at maximal and minimal points of $X_i$}
\label{sec:defin-value-funct}

Due to a larger number of possible cases we can't easily show that all value functions are bounded on $X_i$ as in
\cite{wakker1991additive-RO}. Instead, we can show which functions are bounded on $X^{S_a}_i$. There can be two cases in which we do not
know if $\phi_i$ is bounded on $X^{S_a}_i$. For maximal points these cases are: 
\begin{itemize}
\item $\max X^{S_a}_i = \max X_i$, or 
\item $M^a_i := \max X^{S_a}_i$ and $M^a_iz_{-i} \in X^{S_a}_i$ only if for some $j \neq i$, we have $z_j = \max X_j$. 
\end{itemize}

\begin{lemma}
  \label{lm:maximal-in-xsa}
  Let $i$ be essential on $X^{S_a}$ which has two or more essential variables. Let $M^a_i$ be the maximal element of $X^{S_a}_i$. Then, $\phi_i$
  is bounded on $X^{S_a}_i$ if either:
  \begin{enumerate}
  \item Exists $y$ and $z_{-i}$, such that $y, M^a_iz_{-i} \in X^{S_a}$ and $y \cgeq M_iz_{-i}$, or 
  \item $i$ is in the same interaction clique as variable $j$, for which the first option applies. 
  \end{enumerate}
\end{lemma}

\begin{proof}
  We start from the first case. We would need to construct $y'$ and $z'_{-i}$, such that they do not contain any minimal or maximal points,
  and still $y' \cgeq M^a_iz'_{-i}$.  Assume $z_{-i}$ contain some maximal points. All coordinates can't be maximal by monotonicity. Find
  $\S^a$-minimal $j$, such that $z_j$ is maximal in $X^{S_a}_j$ (in can be also maximal in $X_j$ or not), but for all coordinates $k$, such
  that $j \S^{M^a_iz_{-i}} k$ (note the superscript, point can belong to more than one $X^{S_a}$ in case if some coordinates are in $\E$),
  we have that $z_k$ is not maximal in $X^{S_a}_k$. We can slightly decrease $z_j$, finding $z'_j \cleq_j z_j$, such that
  $M_iz_{-ij}z'_j \in X^{S_a}$. By monotonicity still $y \cgeq M_iz_{-ij}z'_j$. Proceeding in a similar way we can
  construct $z'_{-i}$ not containing any maximal points. If it contains minimal points as well, we can increase them, starting with the
  $\S^a$-maximal one, staying in $X^{S_a}$ and keeping the relation $y \cgeq M_iz_{-ij}z'_j$ (see Lemma \ref{wakker-lemma-11}). Similarly,
  we can replace all maximal and minimal coordinates of $y$.

  Now we have $y' \cgeq M_iz'_{-i}$ and neither $y'$, nor $z'_{-i}$ contain any extreme coordinates. We can therefore conclude, that by
  monotonicity for every $w_i \in X^{S_a}_i$, we have $\phi_i(w_i) < \phi(y') - \sum _{j \neq i} \alpha_j \phi_j(z'_j)$, which shows that
  $\phi_i$ is bounded from above on $X^{S_a}_i$. 

  The second case relies on Lemma \ref{lm:theta-ij-eq}. It is easy to see that if $y$, such that $y \cgeq M_iz_{-i}$, does not exist, we
  must not be able to increase any of coordinates $z_{-i}$, otherwise doing so would give us such $y$ by monotonicity. It also implies, that
  all other essential variables are in the same interaction clique as $i$, otherwise we would be able to change them and again obtain a
  $y$. Finally, it means that $S_a$-maximal coordinate has a maximal value, and relations between all interacting coordinates are $\E$
  (since we are not able neither to decrease, nor to increase any coordinate).  Thus, we can use Lemma \ref{lm:theta-ij-eq} and conclude
  that $\phi_i$ has an upper bound on $X^{S_a}_i$.  

\end{proof}

Having proved Lemma \ref{lm:maximal-in-xsa}, we can now define $\phi(M_i) := \lim _{z_i \to M_i} \phi_i(z_i)$, and $\phi(m_i) := \lim
_{z_i \to m_i} \phi_i(z_i)$ for all $i \in N$.  Assigning values to minimal elements of $X_i$ can be done in a similar manner. Finally we
can prove the final theorem.

\subsection{Global representation on the whole $X$}
\label{sec:glob-repr-whole}

\begin{theorem}
  \label{theo:representation-with-extremes}
For any $x,y \in X$, we have $x \cgeq y \iff \phi(x) \cgeq \phi(y)$.
\end{theorem}
\begin{proof}
  We proceed as in \cite{wakker1991additive-RO} (Lemma 21) with some modifications. For points that do not contain any maximal or minimal
  coordinates, this has been already proved (Section \ref{sec:global-repr-on-x}). Thus, assume that $x$ or $y$ contain maximal or minimal
  coordinates. Let $x \cgeq y$, and let $x \in X^{S_a}, y \in X^{S_b}$. Find $S_a$-minimal $j$ such that $x_j$ is maximal. We can also
  assume that for all $k$, such that $j \S_a k$, we have $j \S^x k$. If this is not the case, then $x$ belongs to several $X^{S_i}$ (by
  definition of these sets), and there must be one where this condition holds. By Lemma \ref{wakker-lemma-11} we can find
  $x'_j: x_j \cgeq_j x'_j$ such that $x'_jx_{-j} \in X^{S_a}$ and still $x'_jx_{-j} \cgt y$. Proceeding like this we get $x'$ which does not
  contain any maximal coordinates, and $x' \cgt y$. We now need to show that $\phi(x') > \phi(y)$.  Similarly, we can replace minimal
  coordinates of $y$, and so it is now required to show that $\phi(x') > \phi(y')$. So we can assume that $x$ has no maximal and $y$ has no
  minimal coordinates.  $x$ must have a non-minimal $X^{S_a}$-essential coordinate, find a $S_a$-maximal one $i$. Again, we can assume that
  for all $k$, such that $j \S_a k$, we have $j \S^x k$. ;By Lemma \ref{wakker-lemma-11} we can decrease it slightly and find
  $x'_i: x_i \cgt x'_i$ and $x'_ix_{-i} \in X^{S_a}$ and still $x' := x'_ix_{-i} \cgt y$. So we need to show now only that
  $\phi(x') \geq \phi(y)$. If we replace all minimal coordinates of $x'$ by non-minimal ones and all maximal coordinates of $y$ by
  non-maximal ones, then by monotonicity the preference between them is not affected, and by Lemma \ref{lm:representation-without-extreme},
  we have $\phi(x') > \phi(y)$. Thus any small increase of minimal and small decrease of maximal values leads to a strict inequality. By
  definition of $\phi_i$ at extreme elements of $X_i$, we have that $\phi(x')$ is the infimum of all such $\phi$-values, and $\phi(y)$ is
  the supremum. Hence, $\phi(x') \cgeq \phi(y)$.  $x \cgeq y \Rightarrow \phi(x) \cgeq \phi(y)$ also implies
  $\phi(x) \geq \phi(y) \Rightarrow x \cgeq y$.

  Now let $\phi(x) > \phi(y)$. $x$ cannot have all it's essential coordinates minimal, so find $S_a$-maximal $j$, such that $x_j$ is not
  minimal. By denserangedness of $\phi_j$, we can find a non-minimal $x'_j : x_j \cgt_j x'_j$ and still $\phi(x'_jx_{-j}) > \phi(y)$. By the
  above argument, we have $x'_jx_{-j} \cgeq y$, and by strict monotonicity we have $x \cgt y$. 
\end{proof}


\section{Uniqueness}
\label{sec:uniqueness}

When defining functions $\phi_i$ we set $\phi_i(r^0_i) = 0 $ and $\phi_i(r^1_i) = 1$ for all $i$. This can be relaxed, in fact for every
interaction clique $A \in N$, we can choose the origin and scaling factor independently. 

Changing the origin alone would not alter the capacity, but changing the scaling factor would. For example, let us define, as previously,
$\phi'_i(r^0_i) = 0 $ for all $i \in N$, but set $\phi'_i(r^1_i) = 1$ for $i \not \in A$ and $\phi'_j(r^1_j) = t_A$ for $j \in A$, so
now $\alpha^{a'}_j = t_A\alpha^a_j$ for some $j \in A$. Accordingly, when normalizing coefficients $\alpha^{a'}_i$ in the additive
representations, we will be dividing by $\sum_{i \not \in A} \alpha^a_i + t_A \sum_{i \in A} \alpha^a_i$. 

It is not hard to see, that the $A-NA$ lemma (Lemma \ref{lm:A-NA}) is still intact - indeed, sums of $\alpha^{a'}_i$ within each clique
remain the same, just scaled by some common factor - $(\sum_{i \not \in A} \alpha^a_i + t_A \sum_{i \in A} \alpha^a_i)$ for $\alpha^{a'}_i, i \not \in A$,
and $t_A(\sum_{i \not \in A} \alpha^a_i + t_A \sum_{i \in A} \alpha^a_i)$ for $\alpha^{a'}_j, j \in A$.

It is also not hard to show that equivalence classes would still have identical values in different $X^{S_a}$ after such operation. The
following Lemmas prepare this.

\begin{lemma}
\label{lm:null-mobius}
  Let $A_1, A_2, ... $ be interaction cliques of $N$. Then, for any $B : B \cap A_i \neq \emptyset, B \cap A_j \neq \emptyset$, we have
  $m(B) = 0$. Also, if for two sets $A_1$ and $A_2$ we have $m(B) = 0$ for all $B : B \cap A_1 \neq \emptyset, B \cap A_2 \neq \emptyset$,
  then coordinates from $A_1$ do not interact with coordinates from $A_2$.
\end{lemma}
\begin{proof}
  It's enough to show this for singletons. Let $i \in A_1$ and $j \in A_2$. We need to show that for any $B : i,j \in B$, we have
  $m(B) = 0$, and vice versa, if $m(B) = 0$ for every such $B$, then $i$ and $j$ do not interact.  Assume that for some such $B$ we have
  $m(B) \neq 0$. Then, we can find $x_{ij}z_{-ij} \in X^{S_a}$, $y_{ij}z_{-ij} \in X^{S_b}$, such that $\alpha^a_k = \alpha^b_k$ for all
  $k \neq i,j$, and $\alpha^a_i \neq \alpha^b_i, \alpha^a_j \neq \alpha^b_j$. This implies that $ij$-trade-off consistency does not hold on
  all $X_{ij}$, hence the variables interact. To show the reverse, note that we have $\alpha^a_i = \alpha^b_i, \alpha^a_j = \alpha^b_j$ for
  all possible points in $X$, which implies $ij$-trade-off consistency on all $X_{ij}$.
\end{proof}
The following two Lemmas follow trivially.
\begin{lemma}
  $\sum _{B \subset A_i} m(B) \geq 0$ for all interaction cliques $A_i$.
\end{lemma}

\begin{lemma}
\label{lm:cliq-add-decomp}
  Let $A_1, A_2, ... $ be interaction cliques of $N$. Then the Choquet integral wrt a corresponding $\nu$,, can be written as a sum of
  integrals wrt ``sub-capacities'', defined on sets of all subsets of $A_i$. 
\end{lemma}

Consequently, we can substitute $\phi_i$ for $k_A\psi_i$ for all $i \in A$, and re-normalize the capacity by multiplying every $m(A)$, hence
also $\nu(A)$ by $\sum _{B \not \in A} m(B) + k_A \sum {B \in A} m(B)$. Apparently, this implies that points from different $X^{S_a}$
belonging to the same equivalence class would still have identical values.

Uniqueness properties of the value functions are similar to those obtained in the homogeneous case $X = Y^n$, but are modified to accommodate for the
heterogeneous structure of the set $X$ in this paper. Because of our general setup, value functions might admit ``ordinal'' transformations
at certain points, and ``cardinal'' at the others, even within the same coordinate. In particular, if a point $x_i$ belongs to some
$X^{S_a}_i$, and $X^{S_a}$ has two or more essential coordinates, one of which is $i$, then $\phi_i(x_i)$ admits only a cardinal
transformation, unless an extreme case applies, when $x_iz_{-i} \in X^{S_a}$ for a single $z_{-i}$.  Two dimensional case is much more
transparent to understanding, and the general idea remains intact. 

Let ${\cal I} = \{A_1, \ldots, A_k \}$ be a set of interaction cliques of $N$. Obviously, $I$ is a partition of $N$. 

\begin{lemma}
  \label{lm:uniqueness-XSA}
  Let $g_1:X_1 \rightarrow \mathbb{R}, \ldots, g_n:X_n \rightarrow \mathbb{R}$ be such that (\ref{eq:repr}) holds with $f_i$ substituted by
  $g_i$. Then, at all $x_i \in X_i$, such that for more than one $z_{-i}$ we have $x_iz_{-i} \in X^{S_a}$, where $X^{S_a}$ has two essential
  coordinates, one of which is $i$, value functions $f_i$ and $g_i$ are related in the following way:
\begin{equation*}
  f_i(x_i) = \alpha_{A_j}g_i(x_i) + \beta_{A_j},
\end{equation*}

Capacity changes as follows
\begin{equation*}
  m'(B) = \frac{\alpha_{A_j}m(B)}{\sum _{C \subset A_i, A_i \in {\cal I}} \alpha_{A_i}m(C)}.
\end{equation*}

At the remaining points of $X$, value functions $f_i$ have the following uniqueness properties:
\begin{equation*}
  f_i(x_i) = \psi_i(g_i(x_i)),
\end{equation*}
where $\psi_i$ is an increasing function, and for all $j \in N, j \neq i$, such that exists $A \in N: i,j \in A, m(A) > 0$, we additionally
have
\begin{equation*}
f_i(x_i) = f_j(x_j) \iff g_i(x_i) = g_j(x_j).
\end{equation*}
\end{lemma}

\begin{proof}
  Mostly follows from uniqueness properties of additive and ordinal representations and Lemma \ref{lm:theta-ij-eq}. The only complication is
  the special case, when $x_iz_{-i} \in X^{S_a}$ for a unique $z_{-i}$. This effectively means that $x_iz_{-i}$ is the only representative
  of its equivalence class in $X^{S_a}$, hence it must be maximal or minimal. It also implies, that the transformation of $\phi_i(x_i)$ does
  not have to be the same as for all other points in $X^{S_a}_i$, as the only condition it has to satisfy is that $\phi(x_iz_{-i})$ has to
  be greater (or less) than values of all other equivalence classes in $X^{S_a}$.
\end{proof}


\section*{Appendix}
\label{cha:appendix}

\renewcommand{\thesection}{A.\arabic{section}}
\setcounter{section}{0}

\section{Technical lemmas}
\label{sec:technical-lemmas}

\begin{lemma}
  \label{lm:3C-independence}
  If $\cgeq$ satisfies triple cancellation then it is independent.
\end{lemma}

\begin{proof}
  $a_ip_{-i} \cleq a_ip_{-i}, a_iq_{-i} \cgeq a_iq_{-i}, a_ip_{-i} \cgeq b_ip_{-i} \Rightarrow a_iq_{-i} \cgeq b_iq_{-i}.$
\end{proof}

\begin{lemma}
  \label{lm:X-unionof-XSa}
  $X = \bigcup X^{S_i}$.
\end{lemma}
\begin{proof}
  Immediate by \textbf{A3}.
\end{proof}

\begin{definition}
  For any set $I$ of consecutive integers, a set $\{g^k_i : g^k_i \in X_i, k \in I \}$ is a \emph{standard sequence} on coordinate $i$ if exist $z_{-ij}, y^0_j, y^1_j$ such that $
  y^0_j \not \sim_j y^1_j $ and for all $i, i+1 \in I$ we have $g^k_iy^1_jz_{-ij} \sim g^{k+1}_iy^0_jz_{-ij}$. Further, we say that $\{g^k_i : k \in I\}$ \emph{is contained in}
  $X^{S_a}$ if $z_{-ij}, y^0_j, y^1_j$ can be chosen in such a way, that all resulting points are in $X^{S_a}$.
\end{definition}

\begin{lemma}
  \label{lm:A5}
  Axiom \textbf{A5} implies the following condition. Let $\{g^k_i : k \in I \}$ and $\{h^k_j : k \in I \}$ be two standard sequences, each contained in some $X^{S_a}$. Assume also,
  that for some $m \in I$, $g^m_iy^0_lz_{-il} \sim h^m_jw^0_nx_{-jn}$ and $g^{m+1}_iy^0_lz_{-il} \sim h^{m+1}_jw^0_nx_{-jn}$. Then, for all $k \in I$, $g^k_iy^0_lz_{-il} \sim
  h^k_jw^0_nx_{-jn}$.
\end{lemma}

\begin{proof}
  The proof is very similar to the one from \citep{krantz1971foundation} (Lemma 5 in section 8.3.1).  Assume wlog that $\{g^k_i: k \in I, g^k_iy^1_lz_{-il} \sim
  g^{k+1}_iy^0_lz_{-il}\}$ is an increasing standard sequence on $i$, which is contained in $X^{S_a}$, whereas $\{h^k_j: h^k_jw^1_nx_{-jn} \sim h^{k+1}w^0_nx_{-jn}\}$ is an
  increasing standard sequence on $j$, and lies entirely in $X^{S_b}$. We will show that $g^{m+2}_iy^0_lz_{-il} \sim h^{m+2}_jw^0_nx_{-jn}$ from which everything follows by
  induction.

  Assume $ h^{m+2}_jw^0_nx_{-jn} \cgt g^{m+2}_iy^0_lz_{-il} $. Then, by restricted solvability exists $\hat{h}_j \in X_j$ such that $\hat{h}_jw^0_nx_{-jn} \sim g^{m+2}_iy^0_lz_{-il}$. By
  \textbf{A5},
    \begin{equation*}
      \left.
      \begin{aligned}
        g^{m}_iy^1_lz_{-il}   & \sim g^{m+1}_iy^0_lz_{-il} \\
        g^{m+1}_iy^1_lz_{-il} & \sim g^{m+2}_iy^0_lz_{-il} \\
        g^m_iy^0_lz_{-il}     & \sim h^m_jw^0_nx_{-jn}     \\ 
        g^{m+1}_iy^0_lz_{-il} & \sim h^{m+1}_jw^0_nx_{-jn} \\
        g^{m+1}_iy^0_lz_{-il} & \sim h^{m+1}_jw^0_nx_{-jn} \\
        g^{m+2}_iy^0_lz_{-il} & \sim \hat{h}_jw^0_nx_{-jn} \\
        h^{m}_jw^1_nx_{-jn}   & \sim h^{m+1}_jw^0_nx_{-jn} 
      \end{aligned}
    \right\} \Rightarrow
    h^{m+1}_jw^1_nx_{-jn} \sim \hat{h}_jw^0_nx_{-jn}.
    \end{equation*}

    By definition of $\{h^k_j\}$, we have $h^{m+1}_jw^1_nx_{-jn} \sim h^{m+2}_jw^0_nx_{-jn} $. Thus, $h^{m+2}_jw^0_nx_{-jn} \sim \hat{h}_jw^0_nx_{-jn}$, hence also
    $g^{m+2}_iy^0_lz_{-il} \sim h^{m+2}_jw^0_nx_{-jn}$, a contradiction. The other cases are symmetrical. 
\end{proof}

\begin{lemma}
  \label{wakker-lemma-11}
Let there be $x,y,i$ such that $x \cgt y$, $x \in X^{S_a}$, $i \S^x k$ for all $k: i \S_a k$, and $x_i$ is non-minimal if $i$ is minimal in
$\S_a$. Then exists $z_i$, such that $z_i \cgeq_i x_i$, $i \S^{z_ix_{-i}} k$ for all $k: i \S_a k$, and $z_ix_{-i} \cgt y$. 
Similarly, let there be $x,y,i$ such that $y \cgt x$, $x \in X^{S_a}$, $k \S^x i$ for all $k: k \S_a i$, and $x_i$ is non-maximal if $i$ is maximal in
$\S_a$. Then exists $z_i$, such that $x_i \cgeq_i z_i$, $k \S^{z_ix_{-i}} i$ for all $k: k \S_a i$, and $y \cgt z_ix_{-i}$. 
\end{lemma}
\begin{proof}
  By restricted solvability and monotonicity. See \cite{wakker1991additive-RO} Lemma 11. 
\end{proof}

\section{Essentiality and monotonicity}
\label{sec:essentiality}

In what follows the essentiality of coordinates within various $X^{S_i}$ is critical. The central mechanism to guarantee consistency in the
number of essential coordinates within various $X^{S_i}$ is bi-independence, which is closely related to comonotonic strong monotonicity of
\cite{wakker1989additive}.

In the Choquet integral representation problem for a heterogeneous product set $X = X_1 \times \ldots \times X_n$, strong monotonicity is
actually a necessary condition. It is directly implied by \textbf{A6} - bi-independence, together with the structural assumption.

\begin{lemma} Pointwise monotonicity.
  \label{lm:4}
  If for all $i \in N$ it holds $x_i \cgeq_i y_i$ , then $x \cgeq y$.
\end{lemma}

\begin{proof}
  $x  \cgeq y_1x_{-1} \cgeq y_{12}x_{-12} \cgeq \ldots \cgeq y$.
\end{proof}

\begin{lemma}
  \label{lm:strong-com-monotonicity}
  If $i$ is essential on $X^{S_a}$, then $a_i \cgeq_i b_i$ iff $a_ix_{-i} \cgt b_ix_{-i}$ for all $a_ix_{-i},b_ix_{-i} \in X^{S_a}$.
\end{lemma}
\begin{proof}
  If $a_i \cgeq_i b_i$ then by the structural assumption exists $a_iz_{-i} \cgt b_iz_{-i}$. The result follows by bi-independence (\textbf{A6}).
\end{proof}

Conceptually, Lemma \ref{lm:strong-com-monotonicity} implies that if a coordinate $i$ is essential on some subset of $X^{S_a}$, then it is
also essential on the whole $X^{S_a}$.  This allows us to make statements like ``coordinate $i$ is essential on $X^{S_a}$''.

\section{Equispaced sequences}
\label{sec:equispaced-sequences}

A usual standard sequence goes along a single dimension as defined above. In this paper we often require to move along several dimensions,
one at a time, maintaining the increment between steps constant in some sense. In order to achieve this we will introduce the concept of
\emph{equispaced} sequences \footnote{See also \citep{bouyssou2010additive} for a similar idea}. Figure \ref{fig:equispaced} illustrates the
process.

\begin{figure}[h!]
  \begin{tikzpicture}[scale=1.5]

    \draw [<->,thick, name path=axis] (0,5) node (yaxis) [left] {$X_j$}
        |- (5.5,0) node (xaxis) [below] {$X_i$};

    \coordinate (r0) at (1,1);
    \draw[dashed, name path=r0xaxis] (r0) -- (5.5,1) coordinate (r0x5);
    \draw[dashed, name path=r0yaxis] (r0) -- (1,4.7) coordinate (r0y5);
    \draw[name path=diag] (0,0) coordinate (a_1) -- (4,4) coordinate (a_2);
    \draw (1.5,0) coordinate (b_1) -- (1,1) coordinate (b_2) -- (0,1.3) coordinate (b_3);
    \coordinate (c) at (intersection of a_1--a_2 and b_1--b_2);

    \fill[black] (r0) circle (1.5pt) node [below left=4pt and -3pt] {$r^0$};

    \draw [dashed] (1.5,0) -- ($(r0)!(b_1)!(r0x5)$);
    \fill [black] ($(r0)!(b_1)!(r0x5)$) circle (1pt);
    \draw let \p1 = ($(r0)!(b_1)!(r0x5)$) in (\x1,1.5pt) -- (\x1,-1.5pt) node[font=\scriptsize, below] {$\alpha^{1}_i$};
    \coordinate (r1) at (3.8,3.8);
    \fill[black] (r1) circle (1.5pt) node [above=4pt] {$r^1$};

    \coordinate (dir1) at ($4*(b_2)-4*(b_1)$);
    \coordinate (dir2) at ($4*(b_3)-4*(b_2)$);

    \foreach  \t in {1.5, 2, ..., 4}
    {
    \coordinate (t) at (\t,0);
    \path[name path=line1] (t) -- +(dir1);
    \draw[name intersections={of=diag and line1,by={Int1}}] (t) -- (Int1);
    \path[name path=line2] (Int1) -- +(dir2);
    \draw[name intersections={of=axis and line2,by={Int2}}] (Int1) -- (Int2);
    \fill [name intersections={of=r0xaxis and line1, by={ssp}}] (ssp) circle (1pt);
  }
  
    \path[name path=line1] (t) -- +(dir1);
    \path[name path=p6] (r1) -| (xaxis -| r1);
    \draw[dashed, name intersections={of=line1 and p6,by={a2_5}}] (r1) -- (a2_5);

    \fill[black] (a2_5) circle (1pt); 
    \draw let \p1 = (a2_5) in (1.5pt,\y1) -- (-1.5pt,\y1) node[font=\scriptsize, left] {$\gamma^{k}_j$};
    \fill [name intersections={of=r0xaxis and line1, by={a_5}}] (a_5) circle (1pt); 
    \draw let \p1 = (a_5) in (\x1,1.5pt) -- (\x1,-1.5pt) node[font=\scriptsize, below] {$\alpha^{k}_i$};


    \coordinate (t) at (4.5,0);
    \path[name path=line1] (t) -- +(dir1);
    \fill [name intersections={of=p6 and line1, by={te1}}] (te1) circle (1pt);
    \draw[name intersections={of=diag and line1,by={Int1}},name path=sl6] (t) -- (Int1);
    \path [name intersections={of=r0xaxis and sl6, by={a_6}}] (a_6) -- (Int1);
    \fill[black] (a_6) circle (1pt);
    \draw let \p1 = (a_6) in (\x1,1.5pt) -- (\x1,-1.5pt) node[font=\scriptsize, below left=0pt and -16pt] {$\alpha^{k+1}_i$};
    \path[name path=line2] (Int1) -- +(dir2);
    \draw[name intersections={of=axis and line2,by={Int2}}] (Int1) -- (Int2);
    
    \coordinate (a2_6) at (intersection of r1--a2_5 and t--Int1);
    \fill[black] (a2_6) circle (1pt);
    \draw let \p1 = (a2_6) in (1.5pt,\y1) -- (-1.5pt,\y1) node[font=\scriptsize, left] {$\gamma^{k+1}_j$};
    \path [name path = ss2_unit] (a2_5)--+(1,0);
    \draw[dashed,name intersections={of=ss2_unit and sl6,by={ss2}}] (a2_5) -- (ss2);
    \draw[dashed, name path = ss2_low] let \p1 = (ss2) in (\x1,4) -- ($(0,0)!(ss2)!(6,0)$);
    \draw[dashed] (a2_5) -- ($(0,0)!(a2_5)!(6,0)$);    
    

    \coordinate (t) at (5,0);
    \path[name path=line1] (t) -- +(dir1);
    \path[name intersections={of=diag and line1,by={Int1}}, name path=sl7] (t) -- (Int1);
    \path[name path=line2] (Int1) -- +(dir2);
    \draw[name intersections={of=axis and line2,by={Int2}}] (Int1) -- (Int2);
    \fill [name intersections={of=p6 and line1, by={te2}}] (te2) circle (1pt);
    \draw let \p1 = (te2) in (1.5pt,\y1) -- (-1.5pt,\y1) node[font=\scriptsize, left] {$\gamma^{k+2}_j$};

    \coordinate (a2_7) at (intersection of r1--a2_5 and t--Int1);
    \path[name path=tint] (t)--(Int1);
    \path [name intersections={of=ss2_low and tint, by=g2_7}];
  
   \draw (a2_7) -- (Int1);
   \draw[name intersections={of=sl7 and r0xaxis, by={a7}}] (a2_7) -- (a7);

\end{tikzpicture}

  \caption{Equispaced sequences in two dimensions}
  \label{fig:equispaced}
\end{figure}
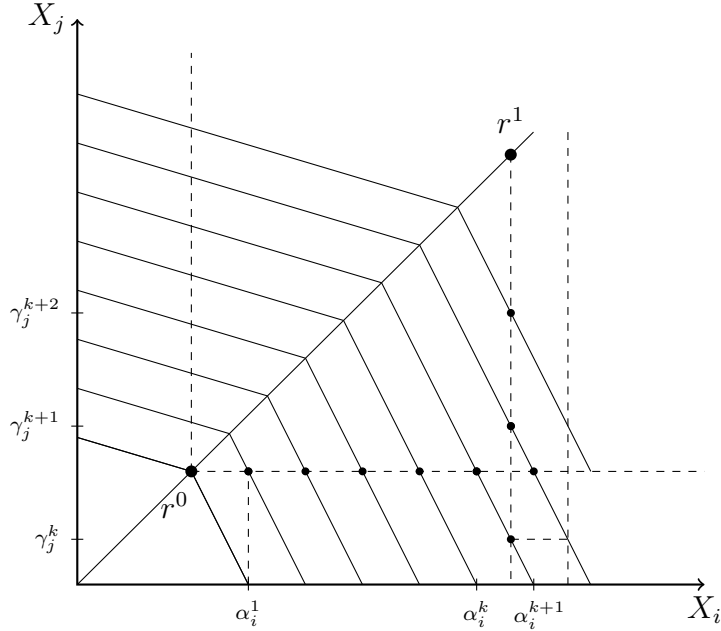

Assume that $r^0, r^1$ are such that $i \R^{r^0} j$ and $i \R^{r^1} j$. We would like to build a sequence from $r^0$ to $r^1$ staying in the
area where $i \R j$. We can choose the size of the sequence step arbitrarily. However, the problem is that $r^1$ does not have an equivalent
point with the second coordinate equal to $r^0_j$, so we cannot build a ``normal'' standard sequence to achieve that. Our aim is to maintain
the sequence within set where $i \R j$. We also assume that $X_i$ and $X_j$ do not have maximal or minimal elements (or they have been removed).

By density and the absence of maximal and minimal elements, we can find $\alpha^k_i, \alpha^{k+1}_i$ such that
$\alpha^{k+1}_i \cgeq_i r^1_i \cgeq_i \alpha^k_i$. We need to change the direction of the sequence from the dimension $i$ to the dimension
$j$ at $r^1_i$.  We construct a point equivalent to $a^k_i$ and a point equivalent to $a^{k+1}_i$ such that their $i$'s coordinate is
$r^1_i$ ( points $\gamma^k_j$ and $\gamma^{k+1}_j$). Since we can choose the step of the sequence arbitrarily, by density and absence of
maximal elements, we can move on and construct a standard sequence on the coordinate $j$ using these two points. 

Remarkably the spacing between subsequent members of the equispaced sequence $\alpha^1,\ldots,\alpha^{k-1},\gamma^k,\gamma^{k+1},\ldots$
stays in a certain sense the same, no matter along which dimension we are moving. Once an additive interval scale is constructed, the
vague notion of the equal spacing will convert into a clear constant difference of values for subsequent members of the sequence.

\paragraph{Extension of A5 to equispaced sequences}

Construction of equispaced sequences allows us to extend the statement of \textbf{A5}, and more precisely that of Lemma \ref{lm:A5} to equispaced
sequences.

\begin{lemma}
  If $g^k$ and $h^k$ are two equispaced sequences entirely lying in $X^{S_a}$ and $X^{S_b}$ correspondingly, and for some $i$ we have
  $g^i \sim h^i$ and $g^{i+1} \sim h^{i+1}$, then for all $j$ such that exist $g^j$ and $h^j$ we have $g^j \sim h^j$.
\end{lemma}
\begin{proof}
  Without loss of generality assume that $g^i := g^i_kg_{-k}$ and $g^{i+1} := g^{i+1}_kg_{-k}$, while $h^i := h^i_lh_{-l}$ and
  $h^{i+1} := h^{i+1}_lh_{-l}$, i.e. in both cases the points are from subsequences on the same dimensions. Assume further, that
  $h^{i+2} := h^{i+2}_lh_{-l}$ and $g^{i+2} := g^{i+2}_mg'_{-m}$, i.e. there is a change of dimension in the sequence $g^k$. We will show
  that this entails that the following steps of the equispaced sequence $g^k$ will be equivalent to the corresponding steps of the sequence
  $h^k$ as long as both keep going along dimensions $m$ and $l$ correspondingly. If either of them changes its direction the technique can
  be applied again. We have by construction (see Figure \ref{fig:equispaced}) $g^{i+2}_kg_{-k} \sim g^{i+2}_mg'_{-m}$ and
  $g^{i+3}_kg_{-k} \sim g^{i+3}_mg'_{-m}$. Therefore, $g^{i+2}_mg'_{-m} \sim h^{i+2}_lh_{-l}$ and $g^{i+3}_mg'_{-m} \sim
  h^{i+3}_lh_{-l}$. The statement follows by Lemma \ref{lm:A5} (\textbf{A5}).
\end{proof}


\section{Necessity of axioms}
\label{sec:necessity-axioms}

\begin{lemma}
  \label{lm:a3-necessity}
  \textbf{A3} is necessary.
\end{lemma}
\begin{proof}
  At any point $z \in X$ and for every $i,j$, we must have either $f_i(z_i) \geq g_i(z_i)$ or $g_i(z_i) \geq f_i(z_i)$. From this everything
  follows trivially (write the condition using Mobius representation of the integral).
\end{proof}

\begin{lemma}
\label{lm:a6-necessity}
  \textbf{A6} is necessary.
\end{lemma}

\begin{proof}
Assume $a_ix_{-i}, b_ix_{-i}, c_ix_{-i}, d_ix_{-i} \in X^{S_a}$ and $a_ix_{-i} \cgt b_ix_{-i}, c_ix_{-i} \sim d_ix_{-i}, c_iy_{-i} \cgt
d_iy_{-i}$. There can be three cases:
\item 
1. $c_iy_{-i},d_iy_{-i} \in X^{S_b}$. We have additive representations on $X^{S_a}$ and $X^{S_b}$, so 
\begin{equation*}
  \begin{aligned}
    \alpha_if_i(a_i) + \sum _{j \in N \setminus i} \alpha_jf_j(x_j) & > \alpha_if_i(b_i) + \sum _{j \in N \setminus i} \alpha_jf_j(x_j) \\
    \alpha_if_i(c_i) + \sum _{j \in N \setminus i} \alpha_jf_j(x_j) & = \alpha_if_i(d_i) + \sum _{j \in N \setminus i} \alpha_jf_j(x_j) \\
    \beta_if_i(c_i) +  \sum _{j \in N \setminus i} \beta_jf_j(y_j)  & > \beta_if_i(d_i) +  \sum _{j \in N \setminus i} \beta_jf_j(y_j). 
  \end{aligned}
\end{equation*}
The first inequality entails $\alpha_i \neq 0$. From this and the following equality follows $f_i(c_i) = f_i(d_i)$, which contradicts with the last inequality. Thus $c_iy_{-i} \cgt
d_iy_{-i}$ implies $c_ix_{-i} \cgt d_ix_{-i}$ but only in the presence of $a_ix_{-i} \cgt b_ix_{-i}$ in the same $X^{S_a}$ (the case when $c_iy_{-i}$ and $d_iy_{-i}$ are not both
in the same $X^{S_b}$ can be reduced to this one. This is also the reason behind the name we gave to this condition - ``weak bi-independence''.
\item
2. $c_iy_{-i},d_iy_{-i} \in X^{S_a}$. In this case we get $\alpha_if_i(c_i) > \alpha_if_i(d_i)$ and $\alpha_if_i(c_i) = \alpha_if_i(d_i)$, a contradiction. 
\item  
3. $c_iy_{-i} \in X^{S_a}, d_iy_{-i} \in X^{S_b}$. As above $\alpha_i \neq 0$, so it follows that $f_i(a_i) = f_i(b_i)$. But then we must
have $d_iy_{-i} \in X^{S_a}$ (value functions are all equal to those for $c_iy_{-i}$), and hence the conclusion follows as in the previous
case.  
\end{proof}

\section{Shape of $\{z_{ij}: i \E^z j\}$}
\label{sec:shape-z_ij:-i}

Shape of the boundary between subsets of $X_{ij}$ where $i \R j$ and $j \R i$ is an interesting and important question. Axiom \textbf{A3}
only guarantees that this boundary is in a certain sense ``quasiconvex'', i.e. an increase along $i$ cannot be matched by a decrease along
$j$. Strengthening this statement requires invoking other axioms.

\subsection{Every $X^{S_i}$ has one essential variable}
\label{sec:single-vari-essent}

Assume that every $X^{S_i}$ has only one essential variable. We will show that an increase along $i$ must be matched by an \emph{increase}
along $j$. This is actually required to construct a representation (see Section \ref{sec:case-with-single}). The main axiom, required to
show this in addition to \textbf{A3} is strong monotonicity (\textbf{A6}). 

\begin{lemma}
\label{lm:E-thin-1ess}
  Let $a_ip_j$ be such that $i \E^{a_ip_jz_{-ij}} j$. Then, unless $i$ and $j$ do not interact, for no $b_i$ we can have $i
  \E^{b_ip_jz_{-ij}} j$. 
\end{lemma}

\begin{proof}
  Assume such $b_i$ exists. Moreover, assume, wlog, that $a_i \cgeq_i b_i$, and we took maximal $a_i$ and minimal $b_i$ for which this
  holds. 
  \begin{enumerate}
  \item If $b_i$ is minimal in $X_i$ and $a_i$ is maximal in $X_i$, then $i,j$ do not interact by \textbf{A3}, so assume that $b_i$ is not
    minimal (the other case is symmetric).
  \item Since we took the smallest $b_i$ for which $i \E^{b_ip_jz_{-ij}} j$, we can assume wlog (other cases are similar) that exists
    $c_i \cleq_i b_i$, such that $b_iq_jz_{-ij} \cgt c_iq_jz_{-ij}$, and $b_ip_jz_{-ij} \sim c_ip_jz_{-ij}$ (a violation of
    $ij$-independence, required by the presence of interaction).
  \item By density assumption there must exist $s_j : i \E^{c_is_jz_{-ij}}$.
  \item We have $b_ip_jz_{-ij} \cgt c_is_jz_{-ij}$, hence $c_ip_jz_{-ij} \cgt c_is_jz_{-ij}$. 
  \item Now we need to extend $X^{S_a} \supset c_iq_jz_{-ij}$ ``to the right'', so that $X^{S_a}_i \supset d_i: d_i \cgt_i b_i$. We also
    need to extend $X^{S_b} \supset b_ip_jz_{-ij}$, so that $X^{S_b}_j \supset t_j: t_j \cgt_j p_j$. This can be done by adjusting
    particular coordinates of $z_{ij}$. Assume, that $z_{ij}$ is already such, that the above conditions hold. In the case when such
    adjustment cannot be performed, we might need to perform a similar extension ``to the left'' from $d_ip_jz_{-ij}$.
  \item Note that on $\SE[c_iq_jz_{-ij}]{ij}$ we can either have $i$ essential, or neither $i$ nor $j$, whereas on $\NW[a_ip_jz_{-ij}]{ij}$ it
    can either be $j$, or none as well.
  \item The point of the extensions is to show that by strong monotonicity \textbf{A6}, we must have $d_ip_jz_{-ij} \cgt b_ip_jz_{-ij}$, as $i$ is essential on
    $\SE[c_iq_jz_{-ij}]{ij}$, but also $d_ip_jz_{-ij} \sim b_ip_jz_{-ij}$, as $i$ is inessential on $\NW[a_ip_jz_{-ij}]{ij}$. If an
    extension ``to the left'' was performed, we must get $c_is_jz_{-ij} \sim b_ip_jz_{-ij}$, but also $ b_ip_jz_{-ij} \cgt c_is_jz_{-ij}$,
    as $j$ stays essential. 
  \end{enumerate}
\end{proof}

\subsection{$X^{S_i}$ have two or more essential variables}
\label{sec:xs_a-has-two}

In section \ref{sec:global-repr-on-x} we have shown, that in the representation value functions for sets $X_i$ and $X_j$ are equal for
points $z$ where $i \E^z j$. Theorem \ref{theo:theta-qual} provides a qualitative version of this statement. The assumption we must make is
that $i$ and $j$ are essential on $X^{S_a}$ and $X^{S_b}$, such that $S_a$ and $S_b$ differ only with respect to order of $i$ and $j$. In
the below proof, we assume that $i,j$ are $S_a$ and $S_b$-maximal, but this can be easily changed, by starting from some $r^1_Ar^0_{-A}$
instead of $r^0$.

\begin{theorem}
  \label{theo:theta-qual}  
  Let $r^0: i \E^{r^1} j, r^1: i \E^{r^1} j$ and $a^k_i$ and $b^k_j$ are two standard sequences such that $a^0_ir^0_{-i} \sim r^0_ir^0_{-i}$
  and $b^0_jr^0_{-j} \sim r^0_jr^0_{-j}$ and $r^1_ir^0_{-i} \sim a^m_ir^0_{-i}$ whereas $r^1_jr^0_{-j} \sim b^m_jr^0_{-j}$.  Assume $r^2$ is
  such that $i \E^{r^2} j$ and $r^2_ir^0_{-i} \sim a^n_ir^0_{-i}$. Then $r^2_jr^0_{-j} \sim b^n_jr^0_{-j}$.
\end{theorem}

\begin{proof}

  Build two equispaced sequences from $r^0$ to $r^1_{ij}r^0_{-ij}$: 
  \begin{itemize}
  \item $e^k$ starting from $r^0_i$ via $r^1_ir^0_{-i}$, and 
  \item $w^k$ starting from $r^0_j$ via $r^1_jr^0_{-j}$,
  \end{itemize}
such that $e^1_ir^0_{-i} \sim w^1_jr^0_{-j}$. By Lemma \ref{lm:A5} (\textbf{A5}) it follows then that all corresponding steps of two
sequences are equivalent, in other words, $e^k \sim w^k$ for all $k$. Consequently, there is the same number of steps both sequences make
between $r^0$ and $r^1_{ij}r^0_{-ij}$, say $K$. 

For some $s < K$ we have $r^1_ir^0_{-i}$ lying between $e^s$ and $e^{s+1}$, i.e. $e^{s+1} \cgt r^1_ir^0_{-i} \cgeq e^s$. Similarly, for some $t <
K$ we have $w^{t+1} \cgt r^1_jr^0_{-j} \cgeq w^t$. We can write:
\begin{equation*}
  [r^0,r^1_ir^0_{-i}] \approx na^k_i \approx se^k,
\end{equation*}
which means: $r^1_ir^0_i$ lies between $a^n_ir^0_{-i}$ and $a^{n+1}_ir^0_{-i}$ and also between $e^s$ and $e^{s+1}$. Similarly, 
\begin{equation*}
  \begin{aligned}[]
    [r^0 , r^1_jr^0_{-j}] & \approx nb^k_j  \approx tw^k \\
    [r^1_ir^0_{-i},r^1_{ij}r^0_{-ij}] & \approx nb^k_j  \approx (K-s)e^k \\
    [r^1_jr^0_{-j},r^1_{ij}r^0_{-ij}] & \approx na^k_i  \approx (K-t)w^k. 
  \end{aligned}
\end{equation*}
Two last statements are possible because by density we can can get arbitrarily close to points $r^1_jr^0_{-j}$ and $r^1_ir^0_{-i}$ by
choosing finer sequences $e^k$ and $w^k$. 

For point $r^2_{ij}r^0_{-ij}$ we have:
\begin{equation*}
  \begin{aligned}[]
    [r^0,r^2_ir^0_{-i}] & \approx ma^k_i  \approx \frac{m}{n}se^k \\
    [r^2_jr^0_{-j},r^2_{ij}r^0_{-ij}] & \approx ma^k_i  \approx \frac{m}{n}(K-t)w^k.
  \end{aligned}
\end{equation*}

Assume that the number of steps on two other segments is different: 

\begin{equation*}
  \begin{aligned}[]
    [r^0,r^2_jr^0_{-j}] & \approx lb^k_j  \approx \frac{l}{n}tw^k \\
    [r^2_ir^0_{-i},r^2_{ij}r^0_{-ij}] & \approx lb^k_j  \approx \frac{l}{n}(K-s)e^k.
  \end{aligned}
\end{equation*}

Summing up parts for both paths to $r^2_{ij}r^0_{-ij}$ we get $\frac{ms + l(K-s)}{n}e^k$ for $\SE{ij}$ and 
$\frac{m(K-t) + lt}{n}w^k$ for $\NW{ij}$. By Lemma \ref{lm:A5} (\textbf{A5}) the number of steps must be identical, so:
\begin{equation*}
  \frac{ms + l(K-s)}{n} = \frac{m(K-t) + lt}{n},
\end{equation*}
or 
\begin{equation*}
  m(s + t - K) = l(s + t - K).
\end{equation*}
There are two possible solutions: 
\begin{itemize}
\item $m=l$, and 
\item $t = K-s$, which means that trade-offs are consistent throughout $X$, hence $i$ and $j$ do not interact, a contradiction.
\end{itemize}

The result follows. 
\end{proof}

\begin{corollary}
\label{lm:horizontal}
  If $a_ip_j: i \E^{a_ip_jz_{-ij}} j$ and $b_ip_j: i \E^{b_ip_jz_{-ij}} j$, then $x_iy_j: i \E^{x_iy_jz_{-ij}} j$ for all $x_i \in X_i$, $y_j \in Y_j$.
\end{corollary}


\begin{corollary}
\label{lm:qual-coord-order}
  If $a_ip_j: i \E^{a_ip_jz_{-ij}} j$, then 
  \begin{itemize}
  \item for any $b_i$, such that $b_i \cgeq_i a_i$ we have $i \S^{b_ip_jz_{-ij}} j$,
  \item for any $b_i$, such that $a_i \cgeq_i b_i$ we have $j \S^{b_ip_jz_{-ij}} i$,
  \item for any $q_j$, such that $q_j \cgeq_j p_j$ we have $j \S^{a_iq_jz_{-ij}} i$,
  \item for any $q_j$, such that $p_j \cgeq_j q_j$ we have $i \S^{a_iq_jz_{-ij}} j$.
  \end{itemize}
 or $i,j$ do not interact.
\end{corollary}

\section{Alternative treatment of the case with single essential variables}
\label{sec:altern-treatm-case}

For this case we can also construct the representation as follows. We will define value functions for all sets $X_i$ in accordance with
$\cgeq_i$. We would additionally require that $\phi_i(x_i) \geq \phi_j(x_j)$ iff $i \R^{x_ix_jz_{-ij}} j$. Finally, we will prove a lemma,
similar to Lemma \ref{lm:A-NA}, which would allow us to construct a capacity and the Choquet integral. 

We start by considering if we can define value functions on $X_i$ according to the rules defined in the previous paragraph. Since $\cgeq_i$
is a weak order, we can obviously define functions such that $\phi_i(x_i) \geq \phi_j(x_j)$ iff $x_i \cgeq_i y_i$. However, the second
condition is more complicated. One particular case, when this would be impossible, is if exist $x_ix_jz_{-ij}$ and $y_ix_jz_{-ij}$ such that
$i \E^{x_ix_jz_{-ij}} j$ and $i \E^{y_ix_jz_{-ij}} j$. However, this would eventually imply that once the representation is constructed we
have $f_i(x_i) = f_i(y_i)$, and hence 
$C(\nu,f(x_iz_{-i})) = C(\nu,f(y_iz_{-i}))$ for any $z_{-i}$, which implies $x_iz_{-i} \sim y_iz_{-i}$ for all $z_{-i}$. This in turn
contradicts a structural assumption that we have made. We state the following lemma:

\begin{lemma}
  Assume that $i$ and $j$ interact. Then, if $i \E^{x_ix_jz_{-ij}} j$, then for no $y_ix_jz_{-ij}$ holds $i \E^{y_ix_jz_{-ij}} j$.
\end{lemma}

\begin{proof}
By \textbf{A3} and strong monotonicity \textbf{A6}. See Lemma \ref{lm:E-thin-1ess} in the appendix for details.
\end{proof}

Next, we will use this to show that for interacting variables $\R$ is transitive. 

\begin{lemma}
  Let $i,j$ and $j,k$ interact. Then, $i \R^z j$ and $j \R^z k$ imply $i \R^z k$.
\end{lemma}

\begin{proof}
  If $i$ and $k$ do not interact, the the  result is immediate, hence, assume that they interact. Also, if we have $i \S^z j$ and $j \S^z$,
  then by acylicity  $i \R^z k$, so one of the relations must be a $\E^z$. 

  Assume first that $i \E^z j, j \S^z k$. Assume also $k \S^z i$. We want to increase $z_i$ slightly, so that for some $x_i \cgeq_i z_i$ we
  have $i \S^{x_iz_{-i}} j$, but still $k \S^{x_iz_{-i}} i$, which leads to a violation of acyclicity. This is possible by density, unless
  $z_i$ is maximal. In this case, decrease $z_j$ slightly, so that for $x_j: z_j \cgeq x_j$ we have $i \S^{x_jz_{-j}} j$ but still
  $i \S^{x_jz_{-j}} k$. This is again possible unless $z_j$ is minimal, in which case we conclude ($z_i$ is maximal, $z_j$ is minimal), that
  $ij-3C$ holds for all $x_i, x_j$, hence $i$ and $j$ do not interact, a contradiction.

  Similarly, assume $i \S^z j, j \E^z k, k \S^z i$. Increase $z_i$ or decrease $z_k$ to violate acyclicity, or otherwise conclude that $j$
  and $k$ do not interact, as above.

  Finally, assume $i \E^z j, j \E^z k, k \S^z i$. Increase $z_i$ and decrease $z_k$ to get $i \S^{x_{ik}z_{-ik}} j, j \S^{x_{ik}z_{-ik}} k,
  k \S^{x_{ik}z_{-ik}} i$, which violates acyclicity. If $z_i$ is maximal, decrease $z_k$ to get $i \E^{x_{k}z_{-k}} j, j \S^{x_{k}z_{-k}} k,
  k \S^{x_{k}z_{-k}} i$, which is the first case considered above. If $z_i$ can be increased, but $z_k$ is minimal, we obtain he second case
  above. Finally, if $z_i$ is maximal and $z_k$ is minimal, then by definition $i \R^z k$. 

  Since $\R$ is transitive it follows that for interacting variables $\E$ and $\S$ would also be transitive. 
\end{proof}

Finally, we can construct value functions $\phi_i:X_i \rightarrow \mathbb{R}$. To do so, we start by constructing the function on any
variable $X_j$. Since $\cgeq_i$ is a weak order, this can obviously be done. Next, we construct the value function on $X_k$, such that $j$
and $k$ interact. This time, the function would also have to satisfy the second constraint - $\phi_k(x_k) \geq \phi_j(x_j)$ iff
$k \R^{x_kx_jz_{-kj}} j$. Then the construction proceeds to some $X_l$, such that $l$ and $j$ or $l$ and $k$ interact. We continue like
this, until all variables in the interaction clique, containing $j$ do not have value functions defined. After this, variables from all
other interaction cliques can be defined in the same way.

Assume, that we have already defined some value functions, in a manner, consistent with the constraints above, and want to define
$\phi_i$. A contradiction can occur, if at some stage we get $i \R^z j$ but $\phi_j(z_j) > \phi_i(z_i)$. This can happen, if $i$ is
interacting with more than one variable, and we used some other value function $\phi_k$ to define $\phi_i$. However, we can update value
functions throughout the clique to resolve this violation.

For any $\phi_i(x_i)$ we have either:
\begin{enumerate}
\item $\phi_i(x_i) = \phi_k(x_k)$ for some $x_k$ in $X_k$, or
\item $\phi_i(x_i) > \phi_k(x_k)$ for all $x_k$ in $X_k$, or
\item $\phi_i(x_i) < \phi_k(x_k)$ for all $x_k$ in $X_k$.
\end{enumerate}
Same holds for every other interacting pair of variables, apart from $i,j$, for which the value functions have already been constructed.

Since we have $\phi_j(z_j) > \phi_i(z_i)$, but $i \R^z j$, in order to resolve the conflict, we need to increase $\phi_i(z_i)$ and decrease
$\phi_j(z_j)$. Since value functions are ordinal, we can do this by applying some increasing transformation to them. If it is the case that
$\phi_i(z_i) > \phi_k(x_k)$ for all $x_k$ in $X_k$, this can be done straightforwardly. However, if it is the case that
$\phi_i(z_i) = \phi_k(z_k)$ or $\phi_i(z_i) < \phi_k(x_k)$ for all $x_k$ in $X_k$, we will have to change $\phi_k$ as well. Eventually,
either the conflict is resolved, i.e. value functions are sufficiently adjusted, or we find some common ``predecessor'' - a coordinate
$m$ and build a chain like the one below: 
\begin{equation*}
  \phi_i(z_i) \leq \phi_k(x_k) \leq \phi_l(x_l) \leq \ldots \leq \phi_o(x_o) \leq \phi_m(x_m),
\end{equation*}
where each point is taken so that it is minimal with respect to the predecessor, for example in pair $\phi_i(z_i) \leq \phi_k(x_k)$ we pick
$x_k$ such that for no $y_k$ we have $\phi_i(z_i) \leq \phi_k(y_k) \leq \phi_k(x_k)$. Similarly, for $\phi_j(z_j)$ we get the following chain:
\begin{equation*}
  \phi_j(z_j) \geq \phi_s(w_s) \geq \phi_t(w_t) \geq \ldots \geq \phi_m(w_m).
\end{equation*}
If now $\phi_m(w_m) > \phi_m(x_m)$, we get 
\begin{equation*}
  \phi_j(z_j) \geq \phi_s(w_s) \geq \phi_t(w_t) \geq \ldots \geq \phi_m(w_m) > \phi_m(x_m) \geq \phi_o(x_o) \geq \ldots \geq \phi_l(x_l) \geq \phi_k(x_k)
  \geq \phi_i(z_i),
\end{equation*}
and hence at $q := z_jw_sw_t \ldots w_m x_o\ldots x_l x_k z_i z_{-jst \ldots mo \ldots lki}$ we have
\begin{equation*}
j \R^q s \R^q t \R^q \ldots \R^q m \S^q o \R^q \ldots \R^q l \R^q \R^q k \R^q i,
\end{equation*}
where every variable interacts with a subsequent one. Hence, as shown previously, we must
have $j \S^q i$, hence $j \S^z i$, which is a contradiction.


\bibliographystyle{abbrvnat}
\bibliography{cite_lib}

\end{document}